\documentclass[lettersize,onecolumn]{IEEEtran}
\usepackage{amsmath,amsfonts}
\usepackage{algorithmic}
\usepackage{algorithm}
\usepackage{array}
\usepackage{textcomp}
\usepackage{url}
\usepackage{verbatim}
\usepackage{graphicx}
\usepackage{cite}
\usepackage{comment}
\usepackage[utf8]{inputenc} 
\usepackage[T1]{fontenc}
\usepackage{ifthen}

\usepackage{mathtools}
\usepackage{cuted}
\usepackage{float}
\usepackage{tabularx,booktabs}
\usepackage{threeparttable}
\usepackage{xcolor}
\usepackage{xspace}
\usepackage{amsthm}
\usepackage{amssymb}
\usepackage{mathrsfs}
\usepackage{subcaption}
\usepackage[inline]{enumitem}

\newcommand{\fDivergence}[2]{D_{f}\left(#1||#2\right)}
\newcommand{\fGammaDivergence}[2]{D_{f}^{\Gamma}\left(#1||#2\right)}
\newcommand{\KLDivergence}[2]{D_{\mathrm{KL}}\left(#1||#2\right)}
\newcommand{\ChiSquareDivergence}[2]{\chi^2\left(#1||#2\right)}
\newcommand{\ConditionalKLDivergence}[3]{D_{\mathrm{KL}}\left(#1||#2|#3\right)}

\newcommand{\GeneralizedCGF}[3]{\Lambda_{#1;#2}\left(#3\right)}
\newcommand{\GammaIPM}[2]{W^{\Gamma}\left(#1,#2\right)}
\newcommand{\GammaBarIPM}[2]{W^{\Bar{\Gamma}}\left(#1,#2\right)}
\newcommand{\Expectation}[2]{\mathbb{E}_{#1}\left[#2\right]}
\newcommand{\RadonNikodym}[2]{\frac{\mathrm{d}#1}{\mathrm{d}#2}}
\newcommand{\ProbabilityKernel}[2]{P_{#1|#2}}
\newcommand{\RealNumber}{\mathbb{R}}
\newcommand{\RealNumberNonnegative}{\mathbb{R}_{+}}
\newcommand{\GeneraliationError}[2]{\mathrm{gen}\left(#1,#2\right)}
\newcommand{\GeneralizationErrorAlgorithmic}{\mathrm{gen}\left(\ProbabilityKernel{W}{Z^n}, \nu, \mu\right)}
\newcommand{\InDistributionGeneralizationGap}{\mathrm{gen}\left(\ProbabilityKernel{W}{Z^n}, \nu\right)}
\newcommand{\PPGeneralizationErrorAlgorithmic}{\mathrm{\widetilde{gen}}\left(\ProbabilityKernel{W}{Z^n}, \nu, \mu\right)}
\newcommand{\LossFunction}[2]{\ell\left(#1,#2\right)}
\newcommand{\RecenteredLossFunction}[2]{\Bar{\ell}\left(#1,#2\right)}
\newcommand{\LipschitzNorm}[1]{\|#1\|_{\mathrm{Lip}}}

\newcommand{\TotalVariation}[2]{\mathrm{TV}\left(#1,#2\right)}
\newcommand{\SquaredHellingerDivergence}[2]{H^2\left(#1||#2\right)}
\newcommand{\HellingerDistance}[2]{H\left(#1||#2\right)}
\newcommand{\fCMI}[4][f]{I_{#1}\left(#2;#3|#4\right)}

\newcommand{\CMI}[3]{I\left(#1;#2|#3\right)}
\newcommand{\DisintegratedfCMI}[4][f]{I_{#1}^{#4}\left(#2;#3\right)}

\makeatletter
\newcommand{\pushright}[1]{\ifmeasuring@#1\else\omit\hfill$\displaystyle#1$\fi\ignorespaces}
\newcommand{\pushleft}[1]{\ifmeasuring@#1\else\omit$\displaystyle#1$\hfill\fi\ignorespaces}
\makeatother

\makeatletter
\DeclareRobustCommand\onedot{\futurelet\@let@token\@onedot}
\def\@onedot{\ifx\@let@token.\else.\null\fi\xspace}
\def\eg{\emph{e.g}\onedot} 
\def\ie{\emph{i.e}\onedot} 
\def\cf{\emph{c.f}\onedot} 
\def\st{\emph{s.t}\onedot}

\def\wrt{w.r.t\onedot}

\def\iid{\emph{i.i.d}\onedot}

\newcommand{\lele}{\textcolor{black}}
\newcommand{\lwl}{\textcolor{black}}
\newcommand{\renjie}{\textcolor{black}} 
\makeatother

\theoremstyle{plain}
\newtheorem{proposition}{Proposition}
\newtheorem{theorem}{Theorem}
\newtheorem{corollary}{Corollary}
\newtheorem{lemma}{Lemma}

\theoremstyle{definition}
\newtheorem{definition}{Definition}
\newtheorem{assumption}{Assumption}

\newtheorem{remark}{Remark}

\begin{document}

\title{An Information-Theoretic Framework for Out-of-Distribution Generalization \\ with Applications to \\ Stochastic Gradient Langevin Dynamics}

\author{Wenliang Liu, Guanding Yu, Lele Wang\textsuperscript{\textsection}, and Renjie Liao\textsuperscript{\textsection}}




\maketitle
\begingroup
\renewcommand\thefootnote{}
\footnotetext{Wenliang Liu and Guanding Yu are with the College of Information Science and Electronic Engineering, Zhejiang University, Hangzhou 310027, China (email: \{liuwenliang, yuguanding\}@zju.edu.cn).}
\footnotetext{Lele Wang and Renjie Liao are with the Department of Electrical and Computer Engineering, University of British Columbia, Vancouver, BC V6T1Z4, Canada (email:\{lelewang, rjliao\}@ece.ubc.ca).}
\footnotetext{Renjie Liao is a faculty member at Vector Institute and a Canada CIFAR AI Chair.}
\renewcommand\thefootnote{\textsection}
\footnotetext{Co-corresponding authors.}
\renewcommand\thefootnote{}
\footnotetext{This work was accepted in part at the 2024 IEEE International Symposium on Information Theory \cite{liu2024information} and the 2024 Canadian Workshop on Information Theory.}
\endgroup

\begin{abstract}
We study the Out-of-Distribution (OOD) generalization in machine learning and propose a general framework that establishes information-theoretic generalization bounds. 
Our framework interpolates freely between Integral Probability Metric (IPM) and $f$-divergence, which naturally recovers some known results (including Wasserstein- and KL-bounds), as well as yields new generalization bounds. 
Additionally, we show that our framework admits an optimal transport interpretation. 
When evaluated in two concrete examples, the proposed bounds either strictly improve upon existing bounds in some cases or match the best existing OOD generalization bounds. 
Moreover, by focusing on $f$-divergence and combining it with the Conditional Mutual Information (CMI) methods, we derive a family of CMI-based generalization bounds, which include the state-of-the-art ICIMI bound as a special instance. 
Finally, leveraging these findings, we analyze the generalization of the Stochastic Gradient Langevin Dynamics (SGLD) algorithm, showing that our derived generalization bounds outperform existing information-theoretic generalization bounds in certain scenarios.
\end{abstract}

\begin{IEEEkeywords}
Generalization, Out-of-Distribution, $f$-Divergence, Integral Probability Metric (IPM), Stochastic Gradient Langevin Dynamics (SGLD)
\end{IEEEkeywords}

\section{Introduction}
\renjie{In machine learning, generalization is the ability of a model to make accurate predictions on unseen data during training. How to analyze and improve the generalization of models are the core subjects of machine learning.}
In the past decades, a series of mathematical tools have been invented or applied to bound the generalization gap, \renjie{\ie, the difference between testing and training performance of a model}, such as the VC dimension \cite{vapnik1971uniform}, Rademacher complexity \cite{bartlett2002rademacher}, covering numbers \cite{pollard1984convergence}, algorithmic stability \cite{bousquet2002stability}, and PAC Bayes \cite{mcallester1998some}. 
Recently, there have been attempts to bound the generalization gap using information-theoretic tools. The idea is to regard the learning algorithm as a communication channel that maps the input set of samples $S$ to the output hypothesis (model weights) $W$. In the pioneering work \cite{russo2016controlling, xu2017information}, the generalization gap is bounded by the Mutual Information (MI) between $W$ and $S$, which reflects the intuition that a learning algorithm generalizes well if \renjie{the learned model weights $W$} leaks little information about the training samples. 
However, the generalization bound becomes vacuous whenever the MI term is infinite. The remedy is to replace the MI term with other smaller quantities to reflect the information that $S$ leaks to $W$, and this has been fulfilled by two orthogonal works, exploiting the Individual Mutual Information (IMI)~\cite{bu2020tightening} and the Conditional Mutual Information (CMI)~\cite{steinke2020reasoning}, respectively. Specifically, the IMI method bounds the generalization gap using the mutual information between $W$ and each individual training datum $Z_i$, rather than the MI between $W$ and the set of whole samples. In certain scenarios, we have infinite MI term but finite IMI term, and thus the resulting IMI bound significantly enhances the MI bound~\cite{bu2020tightening}. Meanwhile, the CMI method studies the generalization through a set of super-samples (also known as ghost samples), a pair of independent and identically distributed (\iid) copies $Z_i^+$ and $Z_i^-$, and then uses a Rademacher random variable $R_i$ to choose $Z_i^+$ or $Z_i^-$ as the $i$-th training datum. The resulting generalization bound is captured by the CMI between $W$ and the samples' identity $R_i$, conditioned on the super-samples. 
Since then, a line of work \cite{haghifam2020sharpened, hellstrom2020generalization, negrea2019information, rodriguez2021random, zhou2022individually, aminian2022learning, chu2024unified} has been proposed to tighten information-theoretic generalization bounds, 
\lwl{and as a remarkable application, these information-theoretic methods have been employed to study the generalization performance of learning algorithms, \eg, the Langevin dynamics and its variants}. 

In practice, it is often the case that the training data suffer from selection biases \renjie{and data distribution shifts with time, \eg, in continual learning}, causing the distribution of test data differs from the training distribution. 
This motivates researchers to study the Out-of-Distribution (OOD) generalization. 
It is common practice to extract invariant features to improve \renjie{the} OOD performance \cite{arjovsky2019invariant}, and \renjie{an OOD generalization theory} was established in~\cite{ye2021towards} via the perspective of invariant features. 
As a sub-field of OOD generalization, the domain adaptation was systematically studied in \cite{ben2006analysis,ben2010theory,sugiyama2007direct,mansour2012multiple}.
In the information-theoretic regime, \lwl{the mismatch between the training and the testing distributions yields an additional penalty term to the OOD generalization bounds, and such penalty can be measured by either the Kullback–Leibler (KL) divergence \cite{wu2020information, masiha2021learning, wang2022information} or the total variation (TV) as well as the Wasserstein distances \cite{wu2020information, wang2022information}.
\lele{In this paper, we establish a universal framework for the OOD generalization analysis, which interpolates between the Integral Probability Metric (IPM) and the $f$-divergence. 
Our OOD generalization bounds include most existing bounds as special cases and are shown to strictly improve upon existing bounds in some cases.}
The framework is expressed in terms of individual datum and thus can be regarded as a natural extension of the classical IMI method.
Meanwhile, we demonstrate that part of these results can be adapted to the CMI settings, which gives rise to an extension and improvement of the cutting-edge ICIMI method to its $f$-divergence variants.
Finally, we demonstrate that our findings lead to several improved generalization bounds of the Stochastic Gradient Langevin Dynamics (SGLD) algorithm. 
To this end, we notice that the common way to derive generalization bounds for SGLD, as what IMI and CMI methods did, is to leverage the chain rule of the KL divergence.
However, such approaches are inapplicable to our methods in general, simply because the chain rule does not hold for general $f$-divergence. 
The remedy is to exploit the subaddivity of the $f$-divergence, which leads to an asymptotically tighter generalization bound for SGLD.}

\subsection{Related Works}
Information-theoretic generalization bounds have been established in the previous work \cite{wu2020information} and \cite{wang2022information}, under the context of transfer learning and domain adaptation, respectively. 
The KL-bounds are derived in \cite{masiha2021learning} utilizing the rate-distortion theory. 
If we ignore the minor difference of models in the generalization bounds, their results can be regarded as natural corollaries of our framework. 
Moreover, \cite{esposito2021generalization} also studied the generalization bounds using $f$-divergence, but it only considered the in-distribution case, and the results are given in \renjie{the} high-probability form. 
Furthermore, both \cite{lugosi2022generalization} and our work uses the convex analysis (Legendre-Fenchel dual) to study the generalization. However, our work restricts the \emph{dependence measure} to $f$-divergence or IPM. 
In contrast, \cite{lugosi2022generalization} did not designate the specific form of the dependence measure, but relied on the strong convexity of the dependence measure. 
This assumption does not hold for IPM and some of the $f$-divergence. 
Besides, \cite{lugosi2022generalization} did not consider the OOD generalization as well. 
In the final writing phase of this paper, we noticed \lele{two independent works~\cite{viallard2024tighter, wang2024f} that use} a similar technique to perform \renjie{the} generalization analysis. 
In particular, \cite{viallard2024tighter} derived PAC-Bayes bounds based on the $(f,\Gamma)$-divergence, which can be regarded as a high-probability counterpart of our paper. 
\cite{wang2024f} focused on the domain adaptation problem using $f$-divergence, and thus their models and methods are partially covered by this paper. 
Readers are referred to \cite{wang2024f} for more specific results under the context of domain adaptation. 
Additionally, \cite{perlaza2024generalization} studied the generalization error at a high level through the method of gaps, a technique for deriving the closed-form generalization error in terms of information-theoretic measures. The $f$-divergence was also employed to regularize the empirical risk minimization algorithm, as explored in \cite{daunas2024equivalence} and the references therein.

\subsection{Contributions}
\begin{enumerate}
    \item We develop a theoretical framework for establishing information-theoretic generalization bounds in the OOD scenarios, which allows us to interpolate freely between IPM and $f$-divergence. 
    In addition to reproducing existing results, such as the generalization bounds based on \renjie{the} Wasserstein distance~\cite{wu2020information, wang2022information} and the KL divergence~\cite{wu2020information, wang2022information, masiha2021learning}, our framework also derives new generalization bounds. 
    \renjie{Notably, our OOD generalization bounds can be applied to in-distribution cases by simply setting the testing distribution equal to the training distribution.}
    \item Our framework is designed to work with the CMI methods for \renjie{bounded loss functions}, and extends the state-of-the-art CMI-based result, known as the Individually Conditional Individual Mutual Information (ICIMI) bound, to a range of $f$-divergence-based ICIMI bound. 
    This enables us to properly select the function $f$ to derive tighter generalization bounds compared to the standard ICIMI bound.
    \item \renjie{We leverage above results to analyze the generalization of the SGLD algorithm for bounded and Lipschitz loss functions. 
    First, we improve upon an existing SGLD generalization bound in~\cite{bu2020tightening} by introducing an additional domination term and extending it to the OOD setting. 
    Next, we employ the squared Hellinger divergence and ICIMI methods to derive an alternative bound, highlighting its advantages under the assumption of without-replacement sampling. 
    Finally, we relax this without-replacement sampling assumption to an asymptotic condition, showing that the resulting generalization bound is tighter than existing information-theoretic bounds for SGLD in the asymptotic regime.}
\end{enumerate}

\subsection{Notation and Organization}
We denote the set of real numbers and the set of non-negative real numbers by $\RealNumber$ and $\RealNumberNonnegative$, respectively. 
Sets, random variables, and their realizations are respectively represented in Calligraphic fonts (\eg, $\mathcal{X}$), uppercase letters (\eg, $X$), and lowercase letters (\eg, $x$).
Let $\mathcal{P}(\mathcal{X})$ be the set of probability distributions over set $\mathcal{X}$ and $\mathcal{M}(\mathcal{X})$ be the set of measurable functions over $\mathcal{X}$. Given $P,Q\in\mathcal{P}(\mathcal{X})$, 
we write $P\perp Q$ if $P$ is singular to $Q$ and $P\ll Q$ if $P$ is absolutely continuous \wrt $Q$. We write $\RadonNikodym{P}{Q}$ as the Radon-Nikodym derivative. 

The rest of this paper is organized as follows. 
Section~\ref{section::models and preliminaries} introduces the problem formulation and some technical preliminaries. In Section~\ref{section::A general theorem}, we propose a general theorem on OOD generalization bounds and illustrate its optimal transport interpretation. Examples and special cases of the general theorem, including both known generalization bounds and new results, are given in Section~\ref{section::corollaries}. In Section~\ref{section::Improvements on the CMI Methods}, we \renjie{restate} and improve the CMI methods for bounding the generalization gap, and the result is applied to the generalization analysis of the SGLD algorithm in Section~\ref{section::Applications on the SGLD Algorithm}. Finally, the paper is concluded in Section~\ref{section::Conclusions}.

\section{Problem Formulation}
\label{section::models and preliminaries}
In this section, we introduce our problem formulation and some preliminaries.

\subsection{Problem Formulation}
\label{subsection::model}
Denote by $\mathcal{W}$ the hypothesis space and $\mathcal{Z}$ the space of data (\ie, input and output pairs). 
We assume training data $(Z_1,\ldots, Z_n)$ are \iid following the distribution $\nu$.
Let $\ell\colon\mathcal{W}\times \mathcal{Z}\to\RealNumberNonnegative$ be the loss function. 
From the Bayesian perspective, our target is to learn a posterior distribution of hypotheses over $\mathcal{W}$, based on the observed data sampled from $\mathcal{Z}$, such that the expected loss is minimized. 
Specifically, we assume the prior distribution $Q_W$ of hypotheses is known at the beginning \renjie{of the learning}. 
Upon observing $n$ samples, $z^n = \left(z_1, \cdots, z_n\right)\in\mathcal{Z}^n$, a learning algorithm outputs one $w\in\mathcal{W}$ through a process like Empirical Risk Minimization (ERM) \cite{vapnik1991principles}. 
The learning algorithm is either deterministic (\eg, gradient descent with fixed hyperparameters) or stochastic (\eg, stochastic gradient descent). 
Thus, the learning algorithm can be characterized by a probability kernel $\ProbabilityKernel{W}{Z^n}$\footnote{Given $z^n\in\mathcal{Z}^n$, $\ProbabilityKernel{W}{Z^n=z^n}$ is a probability measure over $\mathcal{W}$.}, and its output is regarded as one sample from the posterior distribution $\ProbabilityKernel{W}{Z^n=z^n}$. 

In this paper, we consider the OOD generalization setting where the training distribution $\nu$ differs from the testing distribution $\mu$. 
Given a set of samples $z^n$ and the algorithm's output $w$, the incurred generalization gap is
\begin{equation}
\GeneraliationError{w}{z^n} \coloneq \Expectation{\mu}{\LossFunction{w}{Z}} - \frac{1}{n}\sum_{i=1}^{n}\LossFunction{w}{z_i}.
\end{equation}
Finally, we define the generalization gap of the learning algorithm by taking expectation \wrt $w$ and $z^n$, \ie, 
\begin{equation} 
\GeneralizationErrorAlgorithmic \coloneq \mathbb{E}\left[\GeneraliationError{W}{Z^n}\right], \label{equation::algorithmic generalization error}
\end{equation}
where the expectation is \wrt the joint distribution of $(W,Z^n)$, given by $\ProbabilityKernel{W}{Z^n}\otimes\nu^{\otimes n}$.
An alternative approach to defining the generalization gap is to replace the empirical loss in (\ref{equation::algorithmic generalization error}) with the population loss \wrt the training distribution $\nu$, \ie, 
\begin{equation}
    \PPGeneralizationErrorAlgorithmic \coloneq \Expectation{P_W}{\Expectation{\mu}{\LossFunction{W}{Z}} - \Expectation{\nu}{\LossFunction{W}{Z}}},
    \label{equation::PP algorithmic generalization error}
\end{equation}
where $P_W$ denotes the marginal distribution of $W$. By convention, we refer to (\ref{equation::algorithmic generalization error}) as the Population-Empirical (PE) generalization gap and refer to (\ref{equation::PP algorithmic generalization error}) as the Population-Population (PP) generalization gap. In the next two sections, we focus on bounding both the PP and the PE generalization gap using information-theoretic tools.

\subsection{Preliminaries}
\label{subsection::prerequisites}
\begin{definition}[$f$-Divergence \cite{polyanskiy2022information}]Let $f\colon (0,+\infty)\to \RealNumber$ be a convex function satisfying $f(1) = 0$. Given two distributions $P,Q \in \mathcal{P}(\mathcal{X})$, decompose $P = P_c+P_s$, where $P_c\ll Q$ and $P_s \perp Q$. The $f$-divergence between $P$ and $Q$ is defined by
    \begin{equation}
        \fDivergence{P}{Q}\coloneq \Expectation{Q}{f\left(\RadonNikodym{P_c}{Q}\right)} + f'(\infty)P_s(\mathcal{X}),
        \label{equation::f divergence general}
    \end{equation} 
    where $f'(\infty) = \lim_{x\to +\infty}f(x) / x$. If $f$ is super-linear, \ie, $f'(\infty) = +\infty$, then the $f$-divergence has the form of 
    \begin{equation}
        \fDivergence{P}{Q} =  
        \begin{cases}
            \Expectation{Q}{f\left(\RadonNikodym{P}{Q}\right)}, &\text{if}\ P\ll Q, \\
            +\infty, &\text{otherwise}.
        \end{cases}
        \label{equation::f divergence superlinear}
    \end{equation} 
\end{definition}

\begin{definition}[Generalized Cumulant Generating Function (CGF) \cite{birrell2022f, agrawal2021optimal}]Let $f$ be defined as above and $g$ be a measurable function. The generalized cumulant generating function of $g$ \wrt $f$ and $Q$ is defined by
    \begin{equation}
        \GeneralizedCGF{f}{Q}{g} \coloneq \inf_{\lambda \in \RealNumber} \bigl\{\lambda + \Expectation{Q}{f^*(g-\lambda)}\bigr\},
        \label{equation::generalized cumulant generating function}
    \end{equation}
    where $f^*$ represents the Legendre-Fenchel dual of $f$, as
    \begin{equation}
        f^*(y)\coloneq \sup_{x\in\RealNumber}\bigl\{xy - f(x)\bigr\}.
    \end{equation}
\label{definition::generalized CGF}
\end{definition}
\begin{remark}
    Taking $f(x) = x\log x - (x - 1)$ \renjie{in the $f$-divergence} yields the KL divergence\footnote{Here we choose $f$ to be standard, \ie, $f'(1) = f(1) = 0$.}. A direct calculation shows $f^*(y) = e^y - 1$. The infimum is achieved at $\lambda = \log \Expectation{Q}{e^g}$ and thus $\GeneralizedCGF{f}{Q}{g} = \log \Expectation{Q}{e^{g}}$. This means $\GeneralizedCGF{f}{Q}{t(g - \Expectation{Q}{g})}$ degenerates to the classical cumulant generating function of $g$.
    \label{remark::cumulant generating function}
\end{remark}

If we refer to $Q$ as a fixed reference distribution and regard $\fDivergence{P}{Q}$ as a function of distribution $P$, then the $f$-divergence and the generalized CGF form a pair of Legendre-Fenchel dual, as clarified in Lemma~\ref{lemma::variational representation of f divergence}.
\begin{lemma}[Variational Representation of $f$-Divergence \cite{polyanskiy2022information}] 
    \begin{equation}
        \fDivergence{P}{Q} = \sup_{g}\bigl\{\Expectation{P}{g} - \GeneralizedCGF{f}{Q}{g}\bigr\},
        \label{equation::variational representation of f divergence}
    \end{equation}
    where the supreme can be either taken over
    \begin{enumerate}
        \item the set of all simple functions, or
        \item $\mathcal{M}(\mathcal{X})$, the set of all measurable functions, or
        \item $L_Q^{\infty}(\mathcal{X})$, the set of all $Q$-almost-surely bounded functions.
    \end{enumerate}
\label{lemma::variational representation of f divergence}
\end{lemma}
In particular, we recover the Donsker-Varadhan variational representation of KL divergence by combining Remark~\ref{remark::cumulant generating function} and Lemma~\ref{lemma::variational representation of f divergence}:
\begin{equation}
    \KLDivergence{P}{Q} = \sup_{g}\bigl\{\Expectation{P}{g} - \log \Expectation{Q}{e^g}\bigr\}.
\end{equation}

\begin{definition}[$\Gamma$-Integral Probability Metric \cite{muller1997integral}]
    Let $\Gamma\subseteq\mathcal{M}(\mathcal{X})$ be a subset of measurable functions, then the $\Gamma$-Integral Probability Metric (IPM) between $P$ and $Q$ is defined by
    \begin{equation}
        \GammaIPM{P}{Q} \coloneq \sup_{g\in\Gamma}\bigl\{\Expectation{P}{g} - \Expectation{Q}{g}\bigr\}.
    \end{equation}
    \label{Definition::Gamma-IPM}
\end{definition}
Examples of $\Gamma$-IPM include $1$-Wasserstein distance, Dudley metric, and maximum mean discrepancy. In general, 
if $\mathcal{X}$ is a Polish space with metric $\rho$, then the $p$-Wasserstein distance between $P$ and $Q$ is defined through
\begin{equation}
    W_{p}(P,Q) = \Bigl(\inf_{\eta\in\mathcal{C}(P,Q)}\Expectation{(X,Y)\sim \eta}{\rho(X,Y)^p}\Bigr)^{1/p},
\end{equation}
where $\mathcal{C}(P,Q)$ is the set of couplings of $P$ and $Q$. For the special case $p=1$, the Wasserstein distance can be expressed as IPM due to the Kantorovich-Rubinstein Duality
\begin{equation}
    W_1(P,Q) = \sup_{\LipschitzNorm{g}\leq 1}\bigl\{\Expectation{P}{g} - \Expectation{Q}{g}\bigr\},
    \label{equation::Kantorovich-Rubinstein Duality}
\end{equation}
where $\displaystyle \LipschitzNorm{g} \coloneq \sup_{x,y\in\mathcal{X}}\frac{g(x) - g(y)}{\rho(x, y)}$ is the Lipschitz norm of $g$.

\section{Main Results}
\label{section::A general theorem}
In this section, we first propose an inequality regarding the generalization gap in Subsection~\ref{subsection::fundamental inequality}, which leads to one of our main results, a general theorem on the generalization bounds in Subsection~\ref{subsection::a general theorem}. Finally, we show the theorem admits an optimal transport interpretation in Subsection~\ref{subsection::An optimal transport interpretation}.

\subsection{An Inequality on the Generalization Gap}
\label{subsection::fundamental inequality}
In this subsection, we show the generalization gap can be bounded from above using the $\Gamma$-IPM, $f$-divergence, and the generalized CGF. For simplicity, we denote by $P_i = \ProbabilityKernel{W}{Z_i}\otimes \nu$ and $Q = Q_W \otimes \mu$. Moreover, we define the (negative) re-centered loss function as
    $\RecenteredLossFunction{w}{z} \coloneq \Expectation{\mu}{\LossFunction{w}{Z}} - \LossFunction{w}{z}.$
    

\begin{proposition}
    Let $\Bar{\Gamma}\subseteq\mathcal{M}\left(\mathcal{W}\times\mathcal{Z}\right)$ be a class of measurable functions and assume $\Bar{\ell}\in\Bar{\Gamma}$. Then for arbitrary probability distributions $\eta_i\in \mathcal{P}\left(\mathcal{W}\times \mathcal{Z}\right)$ and arbitrary positive real numbers $t_i>0$, $i\in[n]$, we have
    \begin{equation}
        \GeneralizationErrorAlgorithmic\leq \frac{1}{n}\sum_{i=1}^{n}\Bigl(\GammaBarIPM{P_i}{\eta_i}  \\
        + \frac{1}{t_i}\fDivergence{\eta_i}{Q} + \frac{1}{t_i}\GeneralizedCGF{f}{Q}{t_i\RecenteredLossFunction{W}{Z}}\Bigr).
    \label{equation::fundamental inequality}
    \end{equation}    
\label{proposition::fundamental inequality}
\end{proposition}
\begin{proof}[Proof of Proposition~\ref{proposition::fundamental inequality}]
    If $F^*$ is the Legendre dual of some functional $F:\mathcal{X}\to \RealNumber$, then we have
    \begin{equation}
        (tF)^*(x^*) = tF^*\left(\frac{1}{t}x^*\right),
        \label{equation::appendix::legendre dual of tF}
    \end{equation}
    for all $t\in\RealNumberNonnegative$ and $x^*\in \mathcal{X}^*$, the dual space of $\mathcal{X}$.
    Let $Q$ be a fixed reference distribution, $\eta$ be a probability distribution, and $g$ be a measurable function. Combining the above fact with Lemma~\ref{lemma::variational representation of f divergence} yields the following Fenchel-Young inequality:
    \begin{equation}
        \Expectation{\eta}{g}\leq \frac{1}{t}\fDivergence{\eta}{Q} + \frac{1}{t}\GeneralizedCGF{f}{Q}{tg}, t\in\RealNumberNonnegative.
        \label{equation::Fenchel-Young inequality involving t}
    \end{equation}
    As a consequence, we have
        \begin{align}
    \GeneralizationErrorAlgorithmic 
        &= \Expectation{\ProbabilityKernel{W}{Z^n}\otimes\nu^{\otimes n}}{\Expectation{\mu}{\LossFunction{W}{Z}} - \frac{1}{n}\sum_{i=1}^{n}\LossFunction{W}{Z_i}} \\
        &= \frac{1}{n} \sum_{i=1}^n \Expectation{P_i}{\RecenteredLossFunction{W}{Z_i}}
        \label{eq:prop1-eq4}\\
        &\leq \frac{1}{n} \sum_{i=1}^n \Expectation{P_i}{\RecenteredLossFunction{W}{Z_i}} - \Expectation{\eta_i}{\RecenteredLossFunction{W}{Z_i}}  + \frac{1}{t_i}\Bigl(\fDivergence{\eta_i}{Q} + \GeneralizedCGF{f}{Q}{t_i\RecenteredLossFunction{W}{Z_i}}\Bigr)  \label{eq:prop1-eq1}\\
        &\leq \frac{1}{n} \sum_{i=1}^n \sup_{g\in\Bar{\Gamma}}\bigl\{\Expectation{P_i}{g} - \Expectation{\eta_i}{g}
        \bigr\}+ \frac{1}{t_i}\Bigl(\fDivergence{\eta_i}{Q} + \GeneralizedCGF{f}{Q}{t_i\RecenteredLossFunction{W}{Z_i}}\Bigr)  \label{eq:prop1-eq2}\\
        &= \text{RHS of}~\eqref{equation::fundamental inequality}. \nonumber
    \end{align}
    Here, inequality~\eqref{eq:prop1-eq1} follows from \eqref{equation::Fenchel-Young inequality involving t} and inequality~\eqref{eq:prop1-eq2} follows since $\Bar{\ell}\in\Bar{\Gamma}$, and the last equality follows by Definition~\ref{Definition::Gamma-IPM}. 
\end{proof}
Proposition~\ref{proposition::fundamental inequality} has a close relationship with the $(f,\Gamma)$-divergence \cite{birrell2022f}. 
In Appendix~\ref{subsection::appendix::proof of the fundamental inequality}, We provide an alternative proof of Proposition~\ref{proposition::fundamental inequality} using $(f,\Gamma)$-divergence. Furthermore, we show the inequality in Proposition~\ref{proposition::fundamental inequality} is tight in Appendix~\ref{subsection::appendix::tightness of the fundamental inequality}.

\subsection{A Theorem for OOD Generalization}
\label{subsection::a general theorem}
It is common that the generalized CGF $\GeneralizedCGF{f}{Q}{t\Bar{\ell}}$ does not admit an analytical expression, resulting in the lack of closed-form expression in Proposition~\ref{proposition::fundamental inequality}. 
This problem can be remedied by finding a convex upper bound $\psi(t)$ of $\GeneralizedCGF{f}{Q}{t\Bar{\ell}}$, as clarified in Theorem~\ref{theorem::the general theorem}.
Section~\ref{section::corollaries} provides many cases where $\psi$ is quadratic and Theorem~\ref{theorem::the general theorem} can be further simplified.

\begin{theorem}
Let $\Bar{\ell}\in\Bar{\Gamma}\subseteq\mathcal{M}(\mathcal{W}\times\mathcal{Z})$ and $0<b\leq +\infty$. If there exists a continuous convex function $\psi:[0,+\infty)\to[0,+\infty)$ satisfying $\psi(0) = \psi'(0) = 0$ and $\GeneralizedCGF{f}{Q}{t\Bar{\ell}}\leq \psi(t)$ for all $t\in (0, b)$. Then we have
    \begin{equation}
        \GeneralizationErrorAlgorithmic \leq \frac{1}{n}\sum_{i=1}^{n} \inf_{\eta_i\in\mathcal{P}(\mathcal{W}\times\mathcal{Z})}
        \Bigl\{\GammaBarIPM{P_i}{\eta_i}  + (\psi^*)^{-1}\left(\fDivergence{\eta_i}{Q}\right)\Bigr\},
        \label{equation::the general theorem}
    \end{equation}
    \label{theorem::the general theorem}
    where $\psi^*$ denotes the Legendre dual of $\psi$ and $(\psi^*)^{-1}$ denotes the generalized inverse of $\psi^*$.
\end{theorem}
\begin{remark}
    Technically we can replace $\Bar{\ell}$ with $-\Bar{\ell}$ and prove an upper bound of $-\GeneralizationErrorAlgorithmic$ by a similar argument. This result together with Theorem~\ref{theorem::the general theorem} can be regarded as an extension of the previous result~\cite[Theorem 2]{bu2020tightening}. 
    Specifically, the extensions are two-fold. First, \cite{bu2020tightening} only considered the KL-divergence while our result interpolates freely between IPM and $f$-divergence. Second, \cite{bu2020tightening} only considered the in-distribution generalization while our result applies to the OOD generalization, including the case where the training distribution is not absolutely continuous \wrt the testing distribution.
    \label{remark::bound the generalization error from the other side}
\end{remark}

\begin{proof}[Proof of Theorem~\ref{theorem::the general theorem}]
    We first invoke a key lemma.
    \begin{lemma}[Lemma 2.4 in~\cite{boucheron2013concentration}]
        Let $\psi$ be a convex and continuously differentiable function defined on the interval $[0, b)$, where $0 < b \leq +\infty$. Assume that $\psi(0) = \psi'(0) = 0$ and for every $t \geq 0$, let $\psi^*(t) = \sup_{\lambda\in(0,b)}\left\{\lambda t - \psi(\lambda)\right\}$ be the Legendre dual of $\psi$.
        Then the generalized inverse of $\psi^*$, defined by $\left(\psi^*\right)^{-1}(y) \coloneq \inf\left\{t\geq0:\psi^*(t) > y\right\}$, can also be written as 
        \begin{equation}
            \left(\psi^*\right)^{-1}(y) = \inf_{\lambda\in(0,b)}\frac{y+\psi(\lambda)}{\lambda}.
        \end{equation}
        \label{lemma::inverse of psi star}
    \end{lemma}
    As a consequence of Lemma~\ref{lemma::inverse of psi star}, we have 
    \begin{align}
    \GeneralizationErrorAlgorithmic 
        & \leq \frac{1}{n}\sum_{i=1}^{n} \inf_{\eta_i \in \mathcal{P}(\mathcal{W}\times\mathcal{Z}),\;t_i\in\RealNumberNonnegative} \Bigl\{\GammaBarIPM{P_i}{\eta_i}+ \frac{1}{t_i}\fDivergence{\eta_i}{Q} + \frac{1}{t_i}\GeneralizedCGF{f}{Q}{t_i\RecenteredLossFunction{W}{Z}}\Bigr\} \\
        & \leq \frac{1}{n}\sum_{i=1}^{n} \inf_{\eta_i}\inf_{t_i}\biggl\{\GammaBarIPM{P_i}{\eta_i} + \frac{\fDivergence{\eta_i}{Q}+\psi(t_i)}{t_i}\biggr\} \\
        & = \text{RHS of }(\ref{equation::the general theorem}), \nonumber 
    \end{align}
    where the first inequality follows by Proposition~\ref{proposition::fundamental inequality} and the last equality follows by Lemma~\ref{lemma::inverse of psi star}. 
\end{proof}

In general, compared with checking $\Bar{\ell}\in\Bar{\Gamma}$, it is more direct to check that $\ell\in\Gamma$ for some $\Gamma\subseteq\mathcal{M}(\mathcal{W}\times\mathcal{Z})$. 
If so, we can choose\footnote{Note that $\Gamma-\Gamma\neq 0$, it is the set consists of $g - g'$ \st both $g$ and $g'$ belong to $\Gamma$.} $\Bar{\Gamma} = \Gamma - \Gamma$. 
If we further assume that $\Gamma$ is symmetric, \ie, $\Gamma=-\Gamma$, then we have $\Bar{\Gamma} = 2\Gamma$ and thus 
\begin{equation}
    \GammaBarIPM{P_i}{\eta_i} = 2\GammaIPM{P_i}{\eta_i}.
    \label{equation::the general theorem::2 Gamma}
\end{equation}
The following corollary says whenever inserting~\eqref{equation::the general theorem::2 Gamma} into generalization bounds~\eqref{equation::the general theorem}, the coefficient 2 can be removed under certain conditions. 
\begin{corollary}
    Suppose $\ell\in\Gamma\subseteq\mathcal{M}(\mathcal{W}\times\mathcal{Z})$ and $\Gamma$ be symmetric. 
    Let $P_W$ be the distribution of the algorithmic output $W$ and $\mathcal{C}\left(P_W, \cdot\right)\subseteq\mathcal{P}\left(\mathcal{W}\times\mathcal{Z}\right)$ be a class of distributions whose marginal distribution on $\mathcal{W}$ is $P_W$, then we have 
    \begin{equation}
        \GeneralizationErrorAlgorithmic \leq \frac{1}{n}\sum_{i=1}^{n} \inf_{\eta_i\in\mathcal{C}\left(P_W, \cdot\right)} \\
        \Bigl\{\GammaIPM{P_i}{\eta_i}  + (\psi^*)^{-1}\left(\fDivergence{\eta_i}{Q}\right)\Bigr\}.
        \label{equation::the general theorem::Gamma}
    \end{equation}
    \label{corollary::the general theorem::Gamma}
\end{corollary}


\begin{proof}
    By inequality~\eqref{eq:prop1-eq1}, it suffices to prove
    \begin{equation}
        \Expectation{P_i}{\RecenteredLossFunction{W}{Z_i}} - \Expectation{\eta_i}{\RecenteredLossFunction{W}{Z_i}}
        \leq \GammaIPM{P_i}{\eta_i}.
        \label{equation::proof of the general theorem::Gamma::target}
    \end{equation}
    If so, \eqref{equation::the general theorem::Gamma} will follow by exploiting Lemma~\ref{lemma::inverse of psi star} and optimizing over $t_i$ in~\eqref{eq:prop1-eq1}. Since $\eta_i\in\mathcal{C}\left(P_W,\cdot\right)$, the left-hand side of~\eqref{equation::proof of the general theorem::Gamma::target} is exactly $(\Expectation{\eta_i}{\ell} - \Expectation{P_i}{\ell})$. Thus~\eqref{equation::proof of the general theorem::Gamma::target} follows by $\ell\in\Gamma$ and by the symmetry of $\Gamma$.
\end{proof}
\subsection{An Optimal Transport Interpretation of Theorem~\ref{theorem::the general theorem}}
\label{subsection::An optimal transport interpretation}
    
Intuitively, a learning algorithm generalizes well in the OOD setting if the following two conditions hold simultaneously:
\begin{enumerate*}
    \item The training distribution $\nu$ is close to the testing distribution $\mu$.
    \item The posterior distribution $\ProbabilityKernel{W}{Z_i}$ is close to the prior distribution $Q_W$.
\end{enumerate*}
The second condition can be interpreted as the ``algorithmic stability'' and has been studied by a line of work \cite{raginsky2016information, feldman2018calibrating}.
The two conditions together imply that the learning algorithm generalizes well if $P_i$ is close to $Q$. The right-hand side of (\ref{equation::the general theorem}) can be regarded as a characterization of the ``closeness'' between $P_i$ and $Q$. 
Moreover, inspired by \cite{birrell2022f}, we provide an optimal transport interpretation to the generalization bound (\ref{equation::the general theorem}). 
Consider the task of moving (or reshaping) a pile of dirt whose shape is characterized by distribution $Q$, to another pile of dirt whose shape is characterized by $P_i$. 
Decompose the task into two phases as follows. 
During the first phase, we move $Q$ to $\eta_i$ and this yields an $f$-divergence-type transport cost $\left(\psi^*\right)^{-1}\left(\fDivergence{\eta_i}{Q}\right)$, which is a monotonously increasing transformation of $\fDivergence{\eta_i}{Q}$ (see Lemma~\ref{lemma::inverse of psi star}). During the second phase, we move $\eta_i$ to $P_i$ and this yields an IPM-type transport cost $\GammaIPM{P_i}{\eta_i}$. The total cost is the sum of the two-phased costs and is optimized over all intermediate distributions $\eta_i$. 

In particular, we can say more if both $f$ and $\psi$ are super-linear. By assumption, the $f$-divergence is given by~\eqref{equation::f divergence superlinear} and we have $\left(\psi^*\right)^{-1}(+\infty) = +\infty$. This implies we require $\eta_i \ll Q$ to ensure the cost is finite. In other words, $\eta_i$ is a ``continuous deformation'' of $Q$ and cannot assign mass outside the support of $Q$. On the other hand, if we decompose $P_i$ into $P_i = P_{i}^c + P_{i}^s$, where $P_{i}^c\ll Q$ and $P_{i}^s \perp Q$, then all the mass of $P_{i}^s$ is transported during the second phase.

\section{Special Cases and Examples}
\label{section::corollaries}
In this section, we demonstrate how a series of generalization bounds, including both PP-type and PE-type, can be derived through Theorem~\ref{theorem::the general theorem} and its Corollary~\ref{corollary::the general theorem::Gamma}. For simplicity, we defer all the proofs in this section to the Appendix~\ref{section::proofs of the special cases}.

\subsection{Population-Empirical Generalization Bounds}
\label{subsection::P-E generalization bounds}
This subsection bounds the PE generalization gap defined in~\eqref{equation::algorithmic generalization error}. In particular, the PE bounds can be divided into two classes: the IPM-type bounds and the $f$-divergence-type bounds.

\subsubsection{IPM-Type Bounds}
\label{subsubsection::IPM Type Bounds}
From the Bayes perspective, $Q_W$ is the prior distribution of the hypothesis and thus is fixed at the beginning.
Technically, however, the derived generalization bounds hold for arbitrary $Q_W$ and we can optimize over $Q_W$ to further tighten the generalization bounds.
In particular, set $Q_W=P_W$, $\eta_i=Q$, and let $\Gamma$ be the set of $(L_W, L_Z)$-Lipschitz functions. Applying Corollary~\ref{corollary::the general theorem::Gamma} establishes the Wasserstein distance generalization bound.
\begin{corollary}[Wasserstein Distance Bounds for Lipschitz Loss Functions]
    If the loss function is $(L_W, L_Z)$-Lipschitz, i.e., $\ell$ is $L_W$-Lipschitz on $\mathcal{W}$ for all $z\in\mathcal{Z}$ and $L_Z$-Lipschitz on $\mathcal{Z}$ for all $w\in\mathcal{W}$, then we have
    \begin{equation}
        \GeneralizationErrorAlgorithmic \leq L_Z W_1(\nu,\mu) \\
        + \frac{L_W}{n}\sum_{i=1}^{n} \Expectation{\nu}{W_1\left(\ProbabilityKernel{W}{Z_i}, P_W\right)}.
        \label{equation::expectational generalization bounds::Wasserstein}
    \end{equation}%
    \label{corollary::expectational generalization bounds::Wasserstein}%
\end{corollary}
Set $Q_W=P_W$, $\eta_i=Q$, and $\Gamma = \left\{g: a\leq g \leq b\right\}$. Applying Corollary~\ref{corollary::the general theorem::Gamma} establishes the total variation generalization bound. 
\begin{corollary}[Total Variation Bounds for Bounded Loss Function]
    If the loss function is uniformly bounded: $\LossFunction{w}{z}\in [a,b]$, for all $w\in\mathcal{W}$ and $z\in\mathcal{Z}$, then
    \begin{align}    
        \GeneralizationErrorAlgorithmic
        & \leq\frac{b-a}{n}\sum_{i=1}^n \TotalVariation{P_i}{Q} 
        \label{equation::expected generalization bounds::TV::eq-1}\\
        & \leq (b-a)\cdot \TotalVariation{\nu}{\mu}+\frac{b-a}{n}\sum_{i=1}^n \Expectation{\nu}{\TotalVariation{\ProbabilityKernel{W}{Z_i}}{P_W}}.
        \label{equation::expectational generalization bounds::TV::eq-2}
    \end{align}
    \label{corollary::expectational generalization bounds::TV}  
\end{corollary}
Similar results have been \renjie{obtained by others in the context of domain adaptation \cite[Theorem 5.2 and Corollary 5.2]{wang2022information} and transfer learning \cite[Theorem 5 and Corollary 6]{wu2020information}}. In essence, these results are equivalent.

\subsubsection{$f$-Divergence-Type Bounds}
\label{subsubsection::f divergence type bounds}
Set $f(x) = x\log x - (x - 1)$ and $\eta_i = P_i$. For $\sigma$-sub-Gaussian loss functions, we can choose $\psi(t) = \frac{1}{2}\sigma^2 t^2$ and thus $\left(\psi^*\right)^{-1}(y) = \sqrt{2\sigma^2y}$. This recovers the KL-divergence generalization bound \cite{masiha2021learning, wu2020information, wang2022information}.
\begin{corollary}[KL Bounds for sub-Gaussian Loss Functions]
    If the loss function is $\sigma$-sub-Gaussian for all $w\in\mathcal{W}$, we have 
    \begin{equation}
        \GeneralizationErrorAlgorithmic \leq \frac{1}{n}\sum_{i=1}^{n} \sqrt{2\sigma^2\left(I(W;Z_i) + \KLDivergence{\nu}{\mu}\right)},
        \label{equation::expectational generalization bounds::KL sub gaussian}
    \end{equation}
    where $I(W;Z_i)$ is the mutual information between $W$ and $Z_i$.
    \label{corollary::expectational generalization bounds::KL sub gaussian}
\end{corollary}
If the loss function is $(\sigma, c)$-sub-gamma, we can choose $\psi(t) = \frac{t^2}{2(1-ct)}$, $t\in[0, \frac{1}{c})$, and thus $\left(\psi^*\right)^{-1}(y) = \sqrt{2\sigma^2y} + cy$. In particular, the sub-Gaussian case corresponds to $c=0$. 
\begin{corollary}[KL Bounds for sub-gamma Loss Functions]
    If the loss function is $(\sigma,c)$-sub-gamma for all $w\in\mathcal{W}$, we have 
    \begin{equation}
        \GeneralizationErrorAlgorithmic \leq \frac{1}{n}\sum_{i=1}^{n} \sqrt{2\sigma^2\left(I(W;Z_i) + \KLDivergence{\nu}{\mu}\right)} \\
        + c\bigl(I(W;Z_i) + \KLDivergence{\nu}{\mu}\bigr).
    \end{equation}
    \label{corollary::expectational generalization bounds::KL sub gamma}
\end{corollary}
Setting $f(x) = (x-1)^2$ and $\eta_i = P_i$, we establish the $\chi^2$-divergence bound. 
\begin{corollary}[$\chi^2$ Bounds]
    If the variance $\mathrm{Var}_{\mu}\LossFunction{w}{Z}\leq \sigma^2$ for all $w\in\mathcal{W}$, we have 
    \vspace{-.5em}
    \begin{equation}
        \GeneralizationErrorAlgorithmic\leq \frac{1}{n}\sum_{i=1}^{n}\sqrt{\sigma^2\ChiSquareDivergence{P_i}{Q}}.
        \label{equation::chi square generalization bounds}
    \end{equation}
    In particular, by the chain rule of $\chi^2$-divergence, we have
    \begin{equation}
        \GeneralizationErrorAlgorithmic\leq \frac{1}{n}\sum_{i=1}^{n}\sigma \\
        \sqrt{\Bigl(1 + \sup_{z\in\mathcal{Z}}\ChiSquareDivergence{\ProbabilityKernel{W}{Z_i = z}}{Q_W}\Bigr)
            \bigl(1 + \ChiSquareDivergence{\nu}{\mu}\bigr) - 1}.
        \label{equation::chi square generalization bounds::chain rule}
    \end{equation}
    \label{corollary::expectational generalization bounds::Chi Square}
\end{corollary}
In the remaining part of this subsection, we focus on the bounded loss function. Thanks to Theorem~\ref{theorem::the general theorem}, we need a convex upper bound $\psi(t)$ of the generalized CGF $\GeneralizedCGF{f}{Q}{t\Bar{\ell}}$.  
The following lemma says that $\psi(t)$ is quadratic if $f$ satisfies certain conditions.
\begin{lemma}[Corollary 92 in\cite{agrawal2021optimal}]
    Suppose the loss function $\ell(w,z)\in[a, b],\ \forall w\in\mathcal{W},\ z\in \mathcal{Z}$, $f$ is strictly convex and twice differentiable on its domain, thrice differentiable at 1, and
    \begin{equation}
        \dfrac{27f''(1)}{\left(3 - xf'''(1)/f''(1)\right)^{3}} \leq f''(1+x),
        \label{equation::conditions on the f}
    \end{equation}
    for all $x\geq -1$. Then 
    $\GeneralizedCGF{f}{Q}{t\Bar{\ell}}\leq \dfrac{1}{2}\sigma_f^2t^2$, where $\sigma_f = \dfrac{(b-a)}{2\sqrt{f''(1)}}$.
\label{lemma::quadratic psi for bounded loss function}
\end{lemma}
In Table~\ref{table::comparison between f divergences}, we summarize some common $f$-divergence and check whether condition~\eqref{equation::conditions on the f} is satisfied. As a result of Lemma~\ref{lemma::quadratic psi for bounded loss function}, we have the following corollary.
\begin{corollary}
    Let $\LossFunction{w}{z}\in[a,b]$ for all $w\in\mathcal{W}$ and $z\in\mathcal{Z}$. If $f$ satisfies the conditions in Lemma~\ref{lemma::quadratic psi for bounded loss function}, we have
    \begin{equation}
        \GeneralizationErrorAlgorithmic \leq \frac{1}{n}\sum_{i=1}^n\sqrt{2\sigma_f^2\fDivergence{P_i}{Q}}.
        \label{equation::expected generalization bounds::f divergence and bounded losses}
    \end{equation}
    Some common $f$-divergence and the corresponding coefficient $\sigma_f$ are given by Table \ref{table::correspondence between f divergence and the coefficients}.
    \label{corollary::expectational generalization bounds::f-divergence with bounded loss}
\end{corollary}

\begin{table}[t]
\centering
\begin{threeparttable}
  \caption{Comparison Between $f$-Divergences }
  \label{table::comparison between f divergences}
  \centering
  \begin{tabularx}{0.8\textwidth}{
        >{\hsize=0.28\hsize\raggedright\arraybackslash}X
        >{\hsize=0.48\hsize\centering\arraybackslash}X
        >{\hsize=0.18\hsize\centering\arraybackslash}X
        }
    \toprule
    $f$-Divergence & $f(x)$ & Condition (\ref{equation::conditions on the f}) holds? \\
    \midrule
    $\alpha$-Divergence & $\dfrac{x^\alpha - \alpha x + \alpha - 1}{\alpha(\alpha-1)}$
    &  Only for $\alpha\in[-1, 2]$ \\
    \midrule
    $\chi^2$-Divergence & $(x-1)^2$ & Yes  \\
    \midrule
    KL-Divergence & $x\log x - (x-1)$ &  Yes \\
    \midrule
    Squared Hellinger & $(\sqrt{x} - 1)^2$ & Yes   \\ 
    \midrule
    Reversed KL & $-\log x +x-1$ & Yes  \\
    \midrule
    Jensen-Shannon(with parameter $\theta$)& $\theta x \log x - (\theta x + 1 - \theta)\log(\theta x + 1 - \theta)$ &  Yes\\
    \midrule
    Le Cam& $\dfrac{1-x}{2(1+x)}+\dfrac{1}{4}(x-1)$ & Yes  \\
    \bottomrule
  \end{tabularx}
  \begin{tablenotes}
    \item [1] All $f$ in Table \ref{table::comparison between f divergences} are set to be standard, \ie, $f'(1) = f(1) = 0$.
    \item [2] Both the $\chi^2$-divergence and the squared Hellinger divergence are $\alpha$-divergence, up to a multiplicative constant. In particular, we have $\chi^2 = 2D_2$ and $H^2 = \frac{1}{2}D_{1/2}$. The $\theta$-Jensen-Shannon divergence has the form of $D_{\mathrm{JS}(\theta)}(P||Q) = \theta\KLDivergence{P}{R(\theta)}+ (1-\theta)\KLDivergence{Q}{R(\theta)}$, where $R(\theta)\coloneq \theta P + (1-\theta)Q$ and $\theta\in(0,1)$. The classical Jensen-Shannon divergence corresponds to $\theta=1/2$. 
  \end{tablenotes}
\end{threeparttable}
\end{table}
\begin{table}[hb]
    \centering
    \begin{threeparttable}
        \centering
        \caption{Correspondence of $D_f$ and $\sigma_f$}
        \label{table::correspondence between f divergence and the coefficients}
        \begin{tabular}{c|c c c c}
            \toprule
            $D_f$  &  $D_\alpha$ $(\alpha\in[-1, 2])$ & KL & $\chi^2$ & $H^2$  \\
            \hline
            $\sigma_f$  & $\dfrac{b-a}{2}$  & $\dfrac{b-a}{2}$ & $\dfrac{b-a}{2\sqrt{2}}$ & $ \dfrac{b-a}{\sqrt{2}} $\\
            \midrule
            $D_f$  &  Reversed KL & JS$(\theta)$ & Le Cam & \\
            \hline
            $\sigma_f$  & $\dfrac{b-a}{2}$ & $\dfrac{b-a}{2\sqrt{\theta(1-\theta)}}$ & $b-a$ & \\
            \bottomrule
        \end{tabular}
    \end{threeparttable}
\end{table}
We end this subsection with some remarks.
First, since Corollary~\ref{corollary::expectational generalization bounds::TV} also considers the bounded loss function, it is natural to ask whether we can compare~\eqref{equation::expected generalization bounds::TV::eq-1} and~\eqref{equation::expected generalization bounds::f divergence and bounded losses}. The answer is affirmative and we always have 
\begin{equation}
    \TotalVariation{P_i}{Q} \leq \sqrt{2\sigma_f^2\fDivergence{P_i}{Q}}.
    \label{equation::total variation is tighter}
\end{equation}
This Pinsker-type inequality is given by \cite{agrawal2021optimal}. Thus the bound in~\eqref{equation::expected generalization bounds::TV::eq-1} is always tighter than that in~\eqref{equation::expected generalization bounds::f divergence and bounded losses}.
Secondly, as previously mentioned, we optimize $Q_W$ to tighten the bounds.   
In some examples, \eg, the KL-bound, the optimal $Q_W$ is achieved at $P_W$, but it is not always the case, \eg, the $\chi^2$-bound. 
Lastly, all the results derived in this subsection encompass the in-distribution generalization as a special case, by simply setting $\nu=\mu$. 
If we further set $Q_W = P_W$, then we establish a series of in-distribution generalization bounds by simply replacing $\fDivergence{P_i}{Q}$ with $I_f(W; Z_i)$, the $f$-mutual information between $W$ and $Z_i$.


\subsection{Population-Population Generalization Bounds}
\label{subsection::p-p generalization bounds}
By setting $Q_W=P_W$, $\eta_i = P_W\otimes \nu$, and $\Bar{\Gamma}=\{\Bar{\ell}\}$, Theorem~\ref{theorem::the general theorem} specializes to a family of $f$-divergence-type PP generalization bounds. See Appendix~\ref{subsection::appendix::proof of PP generaliation bounds} for proof.
\begin{corollary}[PP Generalization Bounds] Let $\psi$ be defined in Theorem~\ref{theorem::the general theorem}. If $\GeneralizedCGF{f}{Q}{t\Bar{\ell}(W,Z)}\leq \psi(t)$, then we have
    \begin{equation}
        \PPGeneralizationErrorAlgorithmic \leq \left(\psi^*\right)^{-1}\left(\fDivergence{\nu}{\mu}\right).
        \label{equation::pp generalization bounds}
    \end{equation}
    \label{corollary::pp generalization bounds}
\end{corollary}
\vspace{-1.5em}
By Corollary~\ref{corollary::pp generalization bounds}, each $f$-divergence-type PE bound provided in Section~\ref{subsubsection::f divergence type bounds} possesses a PP generalization bound counterpart, with $\fDivergence{P_i}{Q}$ replaced by $\fDivergence{\nu}{\mu}$.
In particular, under the KL case, we recover the results in~\cite[Theorem 4.1]{wang2022information} if the loss function is $\sigma$-sub-Gaussian:
\begin{equation}|\PPGeneralizationErrorAlgorithmic| \leq \sqrt{2\sigma^2\KLDivergence{\nu}{\mu}},
\end{equation}
where the absolute value comes from the symmetry of the sub-Gaussian distribution. 
The remaining PP generalization bounds are summarized in Table~\ref{table::f divergence PP bounds}. 
\begin{table}[ht]
    \centering
    \caption{$f$-Divergence Bounds of the PP Generalization Gap}
    \label{table::f divergence PP bounds}
    \begin{tabular}{c c }
        \toprule
        Assumptions  & PP Generalization Bounds \\
        \midrule
        \midrule
        $\ell$ is $(\sigma, c)$-sub-gamma & $\sqrt{2\sigma^2\KLDivergence{\nu}{\mu}} + c\KLDivergence{\nu}{\mu}$ \\
        \hline
        $\mathrm{Var}_{\mu}\LossFunction{w}{Z}\leq \sigma^2$, $\forall w\in\mathcal{W}$ & $\sqrt{\sigma^2\ChiSquareDivergence{\nu}{\mu}}$\\
        \hline
        $\ell\in[a,b], \alpha\in[-1,2]$ & $(b-a)\sqrt{D_{\alpha}(\nu||\mu) / 2}$\\
        \hline
        $\ell\in[a,b]$ & $(b-a)\sqrt{H^2(\nu||\mu)}$\\
        \hline
        $\ell\in[a,b]$ &  $(b-a)\sqrt{\KLDivergence{\mu}{\nu} / 2}$\\
        \hline
        $\ell\in[a,b]$ &  $(b-a)\sqrt{\frac{D_{\mathrm{JS}(\theta)}(\nu||\mu)}{2\theta(1-\theta)}}$\\
        \hline
        $\ell\in[a,b]$ &  $(b-a)\sqrt{2D_{\mathrm{LC}}(\nu||\mu)}$\\
        \bottomrule
    \end{tabular}
\end{table}
\begin{remark}
    Corollary~\ref{corollary::pp generalization bounds} coincides with the previous result~\cite{agrawal2021optimal}, which studies the optimal bounds between $f$-divergences and IPMs. Specifically, authors in~\cite{agrawal2021optimal} proved $\GeneralizedCGF{f}{Q}{tg} - t\Expectation{Q}{g}\leq \psi(t)$ if and only if $\fDivergence{P}{Q}\geq\psi^*(\Expectation{P}{g} - \Expectation{Q}{g})$. In our context, $g$ is replaced with $\Bar{\ell}$ and thus $\Expectation{Q}{g}=0$. Thus, Corollary~\ref{corollary::pp generalization bounds} can be regarded as an application of the general result \cite{agrawal2021optimal} in the OOD setting.
    \label{remark::PP generalization bounds are examples of the general result}
\end{remark}

\subsection{Examples}
\label{subsection::examples}
\lwl{In this subsection, we examine the above results in the context of two simple ``learning'' problems. 
We demonstrate that our new $\chi^2$-bound surpasses the existing KL-bound in the first example, and the recovered TV-bound achieves the best in the second example.}
\subsubsection{\renjie{Estimating the Gaussian Mean}}
Consider the task of estimating the mean of Gaussian random variables. 
We assume the training sample comes from the distribution $\mathcal{N}(m,\sigma^2)$, and the testing distribution is $\mathcal{N}(m',(\sigma')^2)$. 
We define the loss function as $\LossFunction{w}{z} = (w-z)^2$, then the ERM algorithm yields the estimation $w = \frac{1}{n}\sum_{i=1}^n z_i$. 
Under the above settings, the loss function is sub-Gaussian with parameter $2((\sigma')^2+\sigma^2/n)$, and thus Corollary~\ref{corollary::expectational generalization bounds::KL sub gaussian} and Corollary~\ref{corollary::expectational generalization bounds::Chi Square} apply. 
The known KL-bounds and the newly derived $\chi^2$-bounds are compared in Fig.~\ref{figure::Gauss::in distribution} and Fig.~\ref{figure::Gauss::Out of distribution}, where we set $(m, \sigma^2) = (1,1)$. 
In Fig.~\ref{figure::Gauss::in distribution} the two bounds are compared under the in-distribution setting, \ie, $m'=m$ and $\sigma'=\sigma$. 
\renjie{A rigorous analysis in Appendix~\ref{subsection::appendix::Gaussian and Bernoulli means} shows that both $\chi^2$- and KL-bound decay at the rate $\mathcal{O}(1/\sqrt{n})$, while the true generalization gap decays at the rate $\mathcal{O}(1/n)$. This additional square root comes from the ${(\psi^*)}^{-1}$ term in Theorem~\ref{theorem::the general theorem}.}
Moreover, the KL-bound has the form of $c\sqrt{\log(1+\frac{1}{n})}$ while the $\chi^2$-bound has the form of $c\sqrt{1/n}$. 
Thus the KL-bound is tighter than the $\chi^2$-bound and they are asymptotically equivalent as $n\to\infty$.  
On the other hand, we compare the OOD case in Fig.~\ref{figure::Gauss::Out of distribution}, where we set $m'=1$ and $(\sigma')^2=2$. 
We observe that the $\chi^2$-bound is tighter than the KL-bound at the very beginning. 
In Fig.~\ref{figure::Gauss::d=8}, \renjie{we fix the ``covariance error'', \ie, $\|\mathbf{\Sigma}'-\mathbf{\Sigma}\|_{\mathrm{F}}=(\sigma')^2-\sigma^2=1$, and extend this example to $d$-dimension with $d=8$. We observe that the $\chi^2$-bound is tighter as $n$ is sufficiently large, and more samples are needed in the higher dimensional case for the $\chi^2$-bound to surpass the KL-bound}.
Mathematically, by comparing the $\chi^2$-bound~\eqref{equation::chi square generalization bounds} and the KL-bound~\eqref{equation::expectational generalization bounds::KL sub gaussian}, we conclude that the $\chi^2$-bound will be tighter than the KL-bound whenever $\ChiSquareDivergence{P_i}{Q}<2\KLDivergence{P_i}{Q}$ since the variance of a random variable is no more than its sub-Gaussian parameter. See Appendix~\ref{subsection::appendix::Gaussian and Bernoulli means} for more details.
\begin{figure*}[ht]
\begin{subfigure}{0.32\textwidth}
    \includegraphics[width=\textwidth]{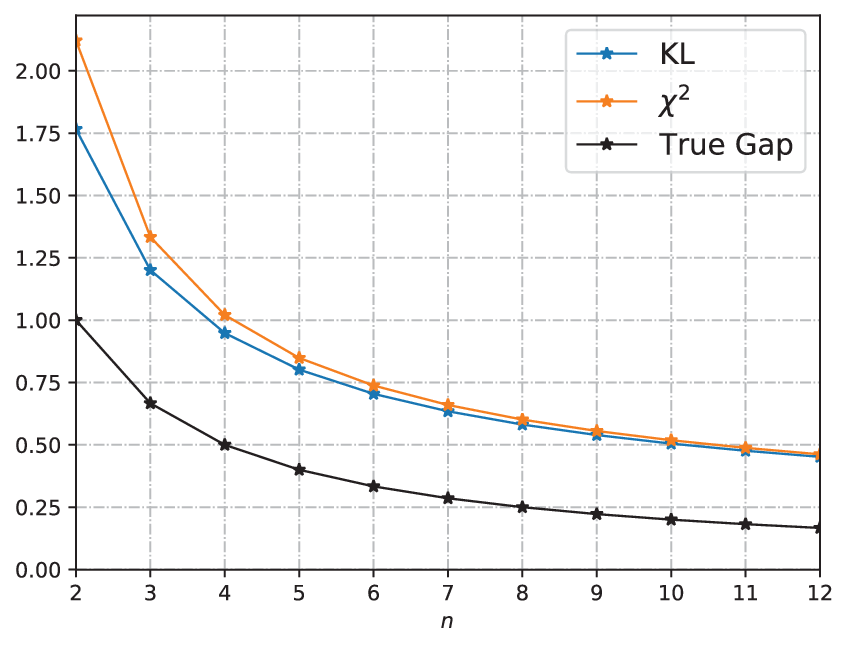}
    \caption{In-distribution, $m=1,\ \sigma^2=1,\ d=1$.}
    \label{figure::Gauss::in distribution}
\end{subfigure}
\hfill
\begin{subfigure}{0.32\textwidth}
    \includegraphics[width=\textwidth]{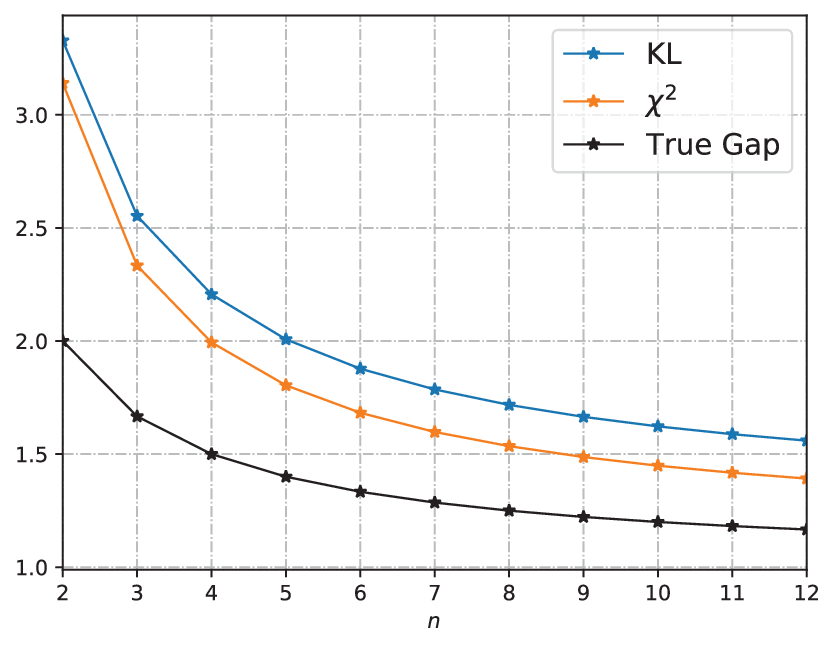}
    \caption{OOD, $m'=1, \ (\sigma')^2=2,\ d=1$.}
    \label{figure::Gauss::Out of distribution}
\end{subfigure}
\hfill
\begin{subfigure}{0.32\textwidth}
    \includegraphics[width=\textwidth]{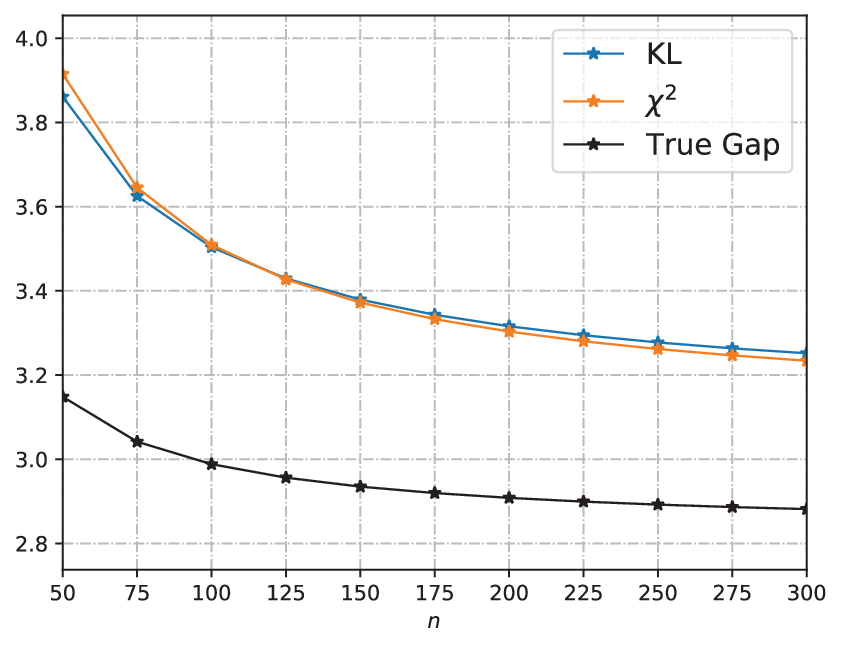}
    \caption{OOD, $\mathbf{m}'=\mathbf{m},\ \|\mathbf{\Sigma}'-\mathbf{\Sigma}\|_\mathrm{F}=1, \ d=8$.}
    \label{figure::Gauss::d=8}
\end{subfigure}
\caption{Generalization Bounds of Estimating Gaussian Means.}
\label{figure:Examples}
\end{figure*}

\subsubsection{\renjie{Estimating the Bernoulli Mean}} 
Consider the previous example where the Gaussian distribution is replaced with the Bernoulli distribution.
We assume the training samples are generated from the distribution $(\mathrm{Bern}(p))^{\otimes n}$ and the test data follows $\mathrm{Bern}(p')$. 
Again we define the loss function as $\LossFunction{w}{z} = (w-z)^2$ and choose the estimation $w = \frac{1}{n}\sum_{i=1}^n z_i$. 

Under the above settings, the loss function is bounded with $\ell\in[0,1]$. If we define the Hamming distance over $\mathcal{W}$ and $\mathcal{Z}$, then the total variation bound coincides with the Wasserstein distance bound. In Fig.~\ref{figure::Bernoulli::p0.3_ptest0.1}, we set $p=0.3$ and $p'=0.1$, and we see that the squared Hellinger, Jensen-Shannon, and Le Cam bounds are tighter than the KL-bound. On the contrary, $\chi^2$- and $\alpha$-divergence bounds are tighter than the KL-bound as in Fig.~\ref{figure::Bernoulli::p0.3_ptest0.5}, where we set $p=0.3$ and $p'=0.5$. But all these $f$-divergence-type generalization bounds are looser than the total variation bound, as illustrated by~\eqref{equation::total variation is tighter}. Additionally, there exists an approximately monotone relationship between $\chi^2$-divergence, $\alpha$-divergence ($\alpha = 3/2$), KL-divergence, and the squared Hellinger divergence. This is because all these bounds are $\alpha$-divergence type, with KL-divergence corresponds to $\alpha = 1$\footnote{Strictly speaking, $D_{\mathrm{KL}} = R_1$, the R\'enyi-$\alpha$-divergence with $\alpha=1$, and $R_\alpha$ is a $\log$-transformation of the $\alpha$-divergence.}. Moreover, we observe that the Le Cam divergence is always tighter than the Jensen-Shannon divergence. This is because the generator $f$ of Le Cam is smaller than that of Jensen-Shannon, and they share the same coefficient $\sigma_f = 1$.
Finally, we consider the extreme case in Fig.~\ref{figure::Bernoulli::n10_p0.6}, where $n=10$, $p=0.6$, and we allow $p'$ decays to $0$. 
When $p'$ is sufficiently small, the KL-bound (along with $\alpha$-divergence ($\alpha=3/2$) and $\chi^2$-bound) is larger than $1$ and thus becomes vacuous. 
\renjie{In contrast,} the squared Hellinger, Jensen-Shannon, Le Cam, and total variation bounds do not suffer from such a problem. 
See Appendix~\ref{subsection::appendix::Gaussian and Bernoulli means} for derivations.

\begin{figure*}[ht]
\centering
\begin{subfigure}{0.32\textwidth}
    \includegraphics[width=\textwidth]{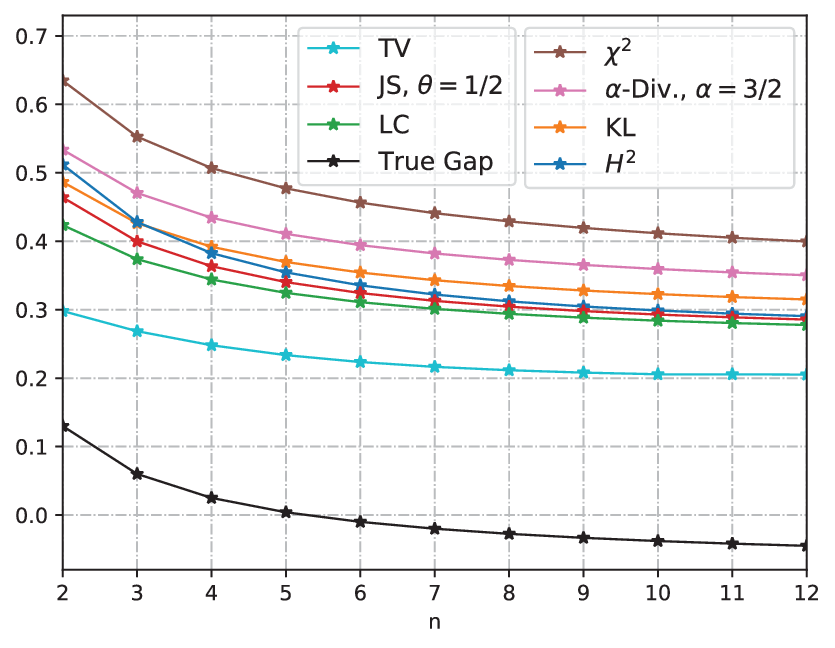}
    \caption{$p=0.3$, $p'=0.1$}
    \label{figure::Bernoulli::p0.3_ptest0.1}
\end{subfigure}
\hfill
\begin{subfigure}{0.32\textwidth}
    \includegraphics[width=\textwidth]{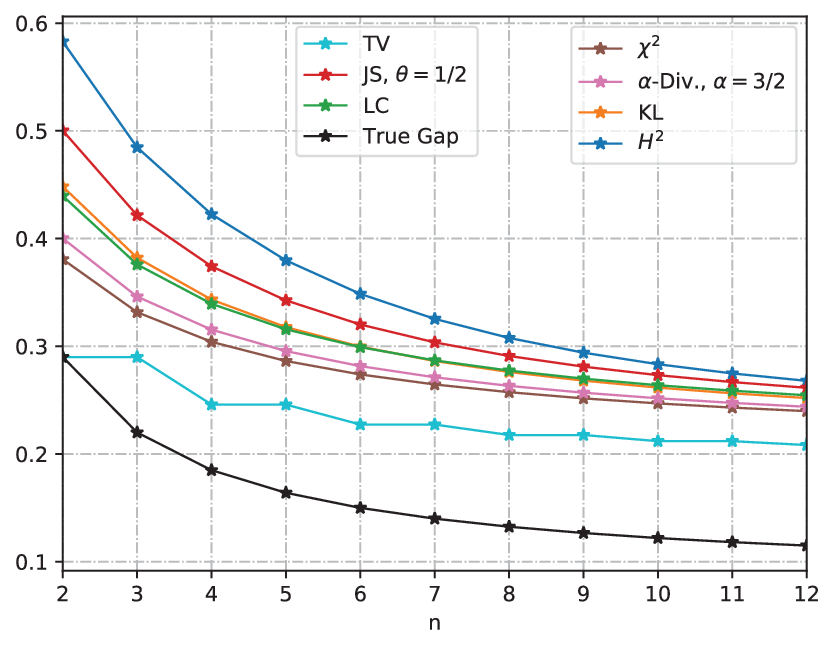}
    \caption{$p=0.3$, $p'=0.5$.}
    \label{figure::Bernoulli::p0.3_ptest0.5}
\end{subfigure}
\hfill
\begin{subfigure}{0.325\textwidth}
    \includegraphics[width=\textwidth]{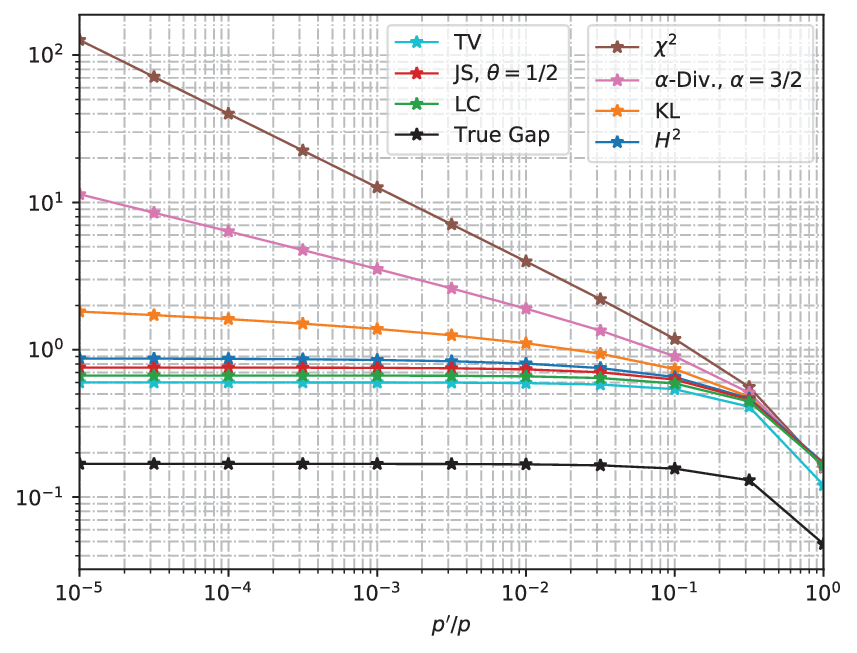}
    \caption{$n=10$, $p=0.6$.}
    \label{figure::Bernoulli::n10_p0.6}
\end{subfigure}
\caption{Generalization Bounds of Estimating Bernoulli Means.}
\label{figure:Examples::Bernoulli}
\end{figure*}

\section{The CMI Method and its Improvement}
\label{section::Improvements on the CMI Methods}

In this section, we study the OOD generalization for bounded loss function using the CMI methods. We first recall some CMI results in Subsection~\ref{subsection::An Overview of CMI Generalization Bounds} and then improve the state-of-the-art result in Subsection~\ref{subsection::f-ICIMI Generalization Bounds} using $f$-divergence. Unless otherwise stated, we assume that $\LossFunction{w}{z}\in[a,b],\ \forall w\in\mathcal{W}, z\in\mathcal{Z}$.

\subsection{\renjie{An Overview of CMI-based In-Distribution Generalization Bounds}}
\label{subsection::An Overview of CMI Generalization Bounds}
Instead of using the training samples $Z^n = (Z_1,\ldots,Z_n)$ generated from distribution $\nu^{\otimes n}$, the CMI method generates a set of super-samples $Z^{n, \pm} = (Z_1^\pm, \ldots,Z_n^\pm)$ from the distribution $\nu^{\otimes 2n}$, where $Z_i^\pm = (Z_i^+, Z_i^-)$ consists of a positive sample $Z_i^+$ and a negative sample $Z_i^-$. When performing training, the CMI framework randomly chooses $Z_i^+$ or $Z_i^-$ as the $i$-th training data $Z_i$ with the help of a Rademacher random variable $R_i$, \ie, $R_i$ takes $1$ or $-1$ with equal probability 1/2. Specifically, we choose $Z_i = Z_i^+$ if $R_i=1$ and choose $Z_i = Z_i^-$ otherwise. 

Based on the above construction, \cite{steinke2020reasoning} established the first in-distribution CMI generalization bound using the mutual information between the output hypothesis $W$ and the identifier $R^n$ conditioned on the super-samples $Z^{n,\pm}$, 
\begin{equation}
    \InDistributionGeneralizationGap \leq (b-a)\sqrt{\frac{2\CMI{W}{R^n}{Z^{n,\pm}}}{n}}.
    \label{equation::CMI bound}
\end{equation}
Here we use $\InDistributionGeneralizationGap$ to denote the expected in-distribution generalization gap by setting $\mu=\nu$. The CMI bound can be improved by combining the `individual' technique in \cite{bu2020tightening}. In particular, the Conditionally Individual Mutual Information (CIMI) bound~\cite{haghifam2020sharpened} is achieved by individualizing only on the data identifier $R_i$, 
\begin{equation}
    \InDistributionGeneralizationGap \leq \frac{(b-a)}{n}\sum_{i=1}^{n}\sqrt{2\CMI{W}{R_i}{Z^{n,\pm}}},
    \label{equation::CIMI bound}
\end{equation}
and the Individually Conditional Individual Mutual Information (ICIMI) bound~\cite{zhou2022individually} is achieved by individualizing on both $R_i$ and the super-sample $Z_i^\pm$,
\begin{equation}
    \InDistributionGeneralizationGap \leq \frac{(b-a)}{n}\sum_{i=1}^{n}\sqrt{2\CMI{W}{R_i}{Z_i^{\pm}}}.
    \label{equation::ICIMI bound}
\end{equation}

It is natural to ask which of these generalization bounds is tighter. If we only consider the mutual information term in these generalization bounds, it follows that $\mathrm{ICIMI} \leq \mathrm{CIMI} \leq \mathrm{CMI} \leq \mathrm{MI}$, and $\mathrm{ICIMI} \leq \mathrm{IMI} \leq \mathrm{MI}$ \cite{zhou2022individually}. This is mathematically written as 
\begin{enumerate}
    \item $\CMI{W}{R_i}{Z_i^\pm} \leq \CMI{W}{R_i}{Z^{n,\pm}} \leq \CMI{W}{R^n}{Z^{n,\pm}} \leq I(W;Z^n)$.
    \item $\CMI{W}{R_i}{Z_i^\pm} \leq I(W;Z_i) \leq I(W;Z^n)$.
\end{enumerate}
When it comes to the generalization bound, it was proved in~\cite{haghifam2020sharpened} that the CIMI bound is tighter than the CMI bound. Thus it follows that the ICIMI bound achieves the best among those CMI-based bounds. However, the ICIMI bound is not necessarily better than the IMI bound. Since a bounded random variable $X\in [a,b]$ is $\frac{b-a}{2}$-sub-Gaussian, the IMI bound reduces to 
\begin{equation}
    \InDistributionGeneralizationGap\leq \frac{b-a}{n}\sum_{i=1}^n\sqrt{\frac{I(W;Z_i)}{2}},
\end{equation}
for bounded loss function $\ell\in[a,b]$. As a consequence, the ICIMI bound is tighter than the IMI bound only if $\CMI{W}{R_i}{Z_i^\pm} < \frac{1}{4}I(W; Z_i)$, which does not necessarily hold in general. 

\subsection{f-ICIMI Generalization Bounds}
\label{subsection::f-ICIMI Generalization Bounds}

We note that all the above generalization bounds were established for the in-distribution case. When extending to the OOD cases, however, the idea of conditioning on super-samples does not work anymore, since the disagreement between training distribution $\nu$ and test distribution $\mu$ breaks the symmetry of the expected generalization gap. To tackle this problem, we decompose the OOD generalization gap as
\begin{align}
    \GeneralizationErrorAlgorithmic &= \Expectation{}{\Expectation{\mu}{\LossFunction{W}{Z}} - \Expectation{\nu}{\LossFunction{W}{Z}} + \Expectation{\nu}{\LossFunction{W}{Z}}- \frac{1}{n}\sum_{i=1}^{n}\LossFunction{W}{Z_i}} \\
    &= \PPGeneralizationErrorAlgorithmic + \InDistributionGeneralizationGap.
    \label{equation::decompose OOD generalization gap as PP gap + in-distribution gap}
\end{align}
\lwl{Here the outer expectation is taken \wrt the joint distribution of $W$ and $Z^n$.}
The first term of~\eqref{equation::decompose OOD generalization gap as PP gap + in-distribution gap} is the PP generalization gap and thus can be bounded using total variation or using $\sqrt{2\sigma_f^2\fDivergence{\nu}{\mu}}$ by Corollary~\ref{corollary::pp generalization bounds}. The second term of~\eqref{equation::decompose OOD generalization gap as PP gap + in-distribution gap} is the in-distribution generalization gap and thus can be bounded by a series of CMI-based methods as introduced above. In particular, by combining the $f$-divergence with the ICIMI method, we have the following improved generalization bound. 
The proof is postponed to the end of this subsection.


\begin{theorem}[$f$-ICIMI Generalization Bound]
    \label{theorem::f-ICIMI}
    Let the loss function $\ell(w,z)\in[a,b], \ \forall w\in\mathcal{W},\ z\in\mathcal{Z}$ and $f$ be an generator of some $f$-divergence satisfying the condition in Lemma~\ref{lemma::quadratic psi for bounded loss function}. Then, we have
    \begin{align}
         \GeneralizationErrorAlgorithmic
        &\leq  (b-a)\TotalVariation{\nu}{\mu} +  \frac{b-a}{n}\sum_{i=1}^n\Expectation{Z^{\pm}_i}{\sqrt{\frac{2\DisintegratedfCMI[f]{W}{R_i}{Z_i^{\pm}}}{f''(1)}}\ } 
        \label{equation::f-ICIMI bound::Expectation is outside the square root}\\
        &\leq  (b-a)\TotalVariation{\nu}{\mu} +  \frac{b-a}{n}\sum_{i=1}^n\sqrt{\frac{2\fCMI[f]{W}{R_i}{Z_i^{\pm}}}{f''(1)}}.
        \label{equation::f-ICIMI bound::Expectation is inside the square root}
    \end{align}
    In the above, $I_f$ denotes the $f$-mutual information and $I_f^z$ denotes the disintegrated $f$-mutual information, \ie, $\DisintegratedfCMI[f]{X}{Y}{z} = \fCMI{X}{Y}{Z=z}$.
\end{theorem}

For the in-distribution case, we recover the ICIMI bound by setting $f = f_{\mathrm{KL}}\coloneqq x\log x$, the generator of KL-divergence. 
Meanwhile, there may exist other $f$ leading to a tighter generalization bound~\eqref{equation::f-ICIMI bound::Expectation is outside the square root} or~\eqref{equation::f-ICIMI bound::Expectation is inside the square root}. 
This is equivalent to \renjie{finding $f$ that satisfies} the condition $\dfrac{\fCMI{W}{R_i}{Z_i^\pm}}{f''(1)}\leq \dfrac{\fCMI[\mathrm{KL}]{W}{R_i}{Z_i^\pm}}{f_{\mathrm{KL}}''(1)}$. 
Nonetheless, the existence of such $f$ is nontrivial, since there is a trade-off between $\fCMI{W}{R_i}{Z_i^\pm}$ and $1 / f''(1)$, see Fig. \ref{figure::graph of generator f}. 
On the one hand, the $f$-mutual information is essentially an expectation of the function $f$, hence a ``steep'' $f$ (like the generator of KL- or $\chi^2$-divergence) would yield a larger $\fCMI{W}{R_i}{Z_i^\pm}$ term. 
On the other hand, the $f''(1)$ reflects the curvature of the graph of $f$ at the point $(1, 0)$, hence a steep $f$ also result a larger $f''(1)$ term. 
Similarly, for those ``flat'' $f$ (like the generator of Jensen Shannon- or squared Hellinger divergence), both the $f$-mutual information and the curvature would become smaller. 
\renjie{This makes selecting $f$ to maximize the quotient $\fCMI{W}{R_i}{Z_i^\pm} / f''(1)$ a nontrivial task.}
\renjie{However, Theorem~\ref{theorem::f-ICIMI} provides more options beyond the KL divergence, allowing us to identify a better choice in certain scenarios, which we will discuss in the next section.}

\begin{figure*}[ht]
    \centering
    \includegraphics[width=0.45\textwidth]{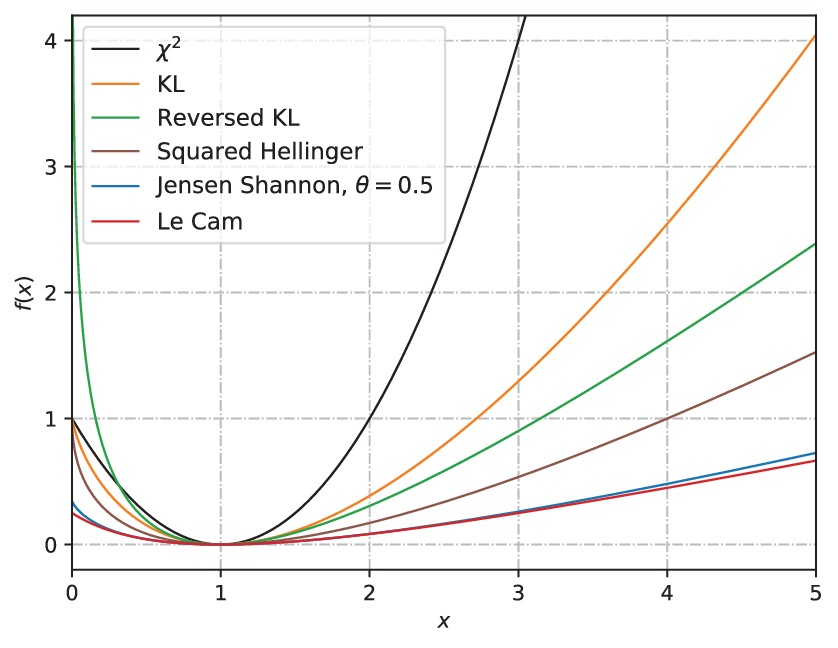}
    \caption{Graph of some common $f$: A tradeoff between $\fCMI{W}{R_i}{Z_i^\pm}$ and $1 / f''(1)$}
    \label{figure::graph of generator f}
\end{figure*}
\begin{proof}[Proof of the Theorem~\ref{theorem::f-ICIMI}]
    By the decomposition~\eqref{equation::decompose OOD generalization gap as PP gap + in-distribution gap}, it suffices to consider the expected in-distribution generalization gap, which can be rewritten as 
    \begin{align}
        \InDistributionGeneralizationGap 
        &= \frac{1}{n}\sum_{i=1}^n 
        \Expectation{Z_i^\pm}{\Expectation{W,R_i|Z_i^\pm}{R_i\left(\LossFunction{W}{Z_i^-} - \LossFunction{W}{Z_i^+}\right)}} \\
        &\leq \frac{1}{n}\sum_{i=1}^n\Expectation{Z_i^\pm}{
        \inf_{t>0}\frac{\fDivergence{\ProbabilityKernel{W,R_i}{Z_i^\pm}}{Q_{W,R_i|Z_i^\pm}}
        + \GeneralizedCGF{f}{Q_{W,R_i|Z_i^\pm}}{tR_i\left(\LossFunction{W}{Z_i^-} - \LossFunction{W}{Z_i^+}\right)}
        }{t}}.
    \end{align}
Choose $Q_{W,R_i|Z_i^\pm} = \ProbabilityKernel{W}{Z_i^\pm}\ProbabilityKernel{R_i}{Z_i^\pm}$ and we have $\fDivergence{\ProbabilityKernel{W,R_i}{Z_i^\pm}}{Q_{W,R_i|Z_i^\pm}} = \DisintegratedfCMI[f]{W}{R_i}{Z_i^\pm}$.
By assumption $\ell\in [a,b]$, we have $R_i\left(\LossFunction{W}{Z_i^-} - \LossFunction{W}{Z_i^+}\right) \in [-(b-a), b-a]$, and thus the Lemma~\ref{lemma::quadratic psi for bounded loss function} follows that
\begin{equation}
    \GeneralizedCGF{f}{Q_{W,R_i|Z_i^\pm}}{tR_i\left(\LossFunction{W}{Z_i^-} - \LossFunction{W}{Z_i^+}\right)} \leq \frac{1}{2f''(1)}(b-a)^2t^2.
\end{equation}
As a consequence, we have
\begin{align}
    \InDistributionGeneralizationGap 
    &\leq \frac{1}{n}\sum_{i=1}^{n}\Expectation{Z_i^\pm}{\inf_{t>0}\frac{1}{t}\left(\DisintegratedfCMI{W}{R_i}{Z_i^\pm} + \frac{(b-a)^2t^2}{2f''(1)}\right)} \\
    &= \frac{(b-a)}{n}\sum_{i=1}^{n}\Expectation{Z_i^\pm}{\sqrt{\frac{2\DisintegratedfCMI[f]{W}{R_i}{Z_i^\pm}}{f''(1)}}\ } \\
    &\leq \frac{(b-a)}{n}\sum_{i=1}^{n}\sqrt{\frac{2\fCMI[f]{W}{R_i}{Z_i^\pm}}{f''(1)}},
\end{align}
where the last inequality follows by Jensen's inequality.
\end{proof}

\section{\renjie{Applications to the SGLD Algorithm}}
\label{section::Applications on the SGLD Algorithm}

\lwl{In this section, we demonstrate an application of preceding results to the stochastic gradient Langevin dynamics (SGLD) algorithm \cite{welling2011bayesian}. 
\renjie{SGLD is a learning algorithm that combines the stochastic gradient descent (SGD), a Robbins–Monro optimization algorithm \cite{robbins1951stochastic}, and Langevin dynamics, a mathematical extension of molecular dynamics models.
Essentially, each iteration of SGLD is an SGD update with added isotropic Gaussian noise.
The noise enables SGLD to escape local minima and asymptotically converge to global minima for sufficiently regular non-convex loss functions.
}
The wide application of SGLD in machine learning has prompted a series of study on its generalization performance~\cite{bu2020tightening,zhou2022individually,negrea2019information,haghifam2020sharpened,rodriguez2021random}.}
The remaining part of this section is organized as follows. 
After a brief introduction to the SGLD in~\ref{subsection::SGLD Introduction}, we exhibit improved SGLD generalization bounds in the following three subsections. 
Motivated by the observation in~\ref{subsection::examples} that TV-based bound achieves the tightest one, we establish a refined SGLD generalization bound in~\ref{subsection::Generalization Bounds on the SGLD Algorithm} and show its advantages whenever the Lipschitz constant of the loss function is large. 
Next, we turn to apply the CMI methods to SGLD. In~\ref{subsection::SGLD using Without Replacement Sampling}, an $f$-ICIMI type generalization bound for SGLD is established on the basis of squared Hellinger divergence. 
Due to the nature of the chain rule of the squared Hellinger divergence, this generalization bound is superior under the without-replacement sampling scheme, and may become worse whenever replacement is allowed. 
The latter problem is remedied in~\ref{subsection::SGLD with big O assumption}, where we utilize a new technique, the sub-additivity of the $f$-divergence, to establish an asymptotically tighter generalization bound for SGLD.

\subsection{Stochastic Gradient Langevin Dynamics}
\label{subsection::SGLD Introduction}
The SGLD algorithm starts with an initialization $W_0$ and iteratively updates $W_t$ through a noisy and stochastic gradient descent process. Mathematically, we write 
\begin{equation}
    W_t = W_{t-1} - \eta_t\nabla_{W}\ell(W_{t-1},Z_{U_t}) + \sigma_t \xi_t,\ t\in[T],
    \label{equation::SGLD update rule}
\end{equation}
where $\xi_t\sim\mathcal{N}(\mathbf{0},\mathbf{I})$ is the standard Gaussian random vector and $\sigma_t$ is the scale parameter. In the standard SGLD setting, we have $\sigma_t = \sqrt{2\eta_t\beta^{-1}}$ where $\beta$ denotes the inverse temperature. Additionally, $U_t\in[n]$ specifies which data is employed for calculating the gradient within the $t$-th iteration (\ie, $U_t = i$ if $Z_i$ is taken), and finally the algorithm outputs $W_T$ after $T$ iterations. We say $U_{[T]} = (U_1,\ldots, U_T)$ is a sample path, which is randomly generated from some process like random sampling over $[n]$ with replacement. Let $\mathcal{T}_i$ be the set of \renjie{steps} $t$ that $Z_i$ is taken for updating $W_t$. Then, $\mathcal{T}_i$ is a deterministic function of the sample path $U_{[T]}$ and we write by $\mathcal{T}_i = \mathcal{T}_i(U_{[T]})$.

The generalization of the SGLD algorithm was also studied by~\cite{bu2020tightening,zhou2022individually}, which we will discuss later. 
Other generalization analyses of the non-stochastic Langevin dynamic and the mini-batch SGLD are discussed in~\cite{negrea2019information,haghifam2020sharpened,rodriguez2021random}, where they use a data-dependent estimation on the prior distribution and their results are incomparable to ours.

\subsection{Generalization Bounds on SGLD for Bounded and Lipschitz Loss Functions}
\label{subsection::Generalization Bounds on the SGLD Algorithm}

Authors in \cite{bu2020tightening} utilized IMI to proved that, if the loss function is $\sigma$-sub-Gaussian and $L$-Lipschitz, the in-distribution generalization gap of the SGLD algorithm can be bounded by
\begin{equation}
    \InDistributionGeneralizationGap \leq \frac{\sigma}{n}\sum_{i=1}^n\Expectation{U_{[T]}}{\sqrt{\sum_{t\in\mathcal{T}_i(U_{[T]})}\frac{\eta_t^2L^2}{\sigma_t^2}}}.
    \label{equation::Bu et al IMI SGLD generalization bound}
\end{equation}
When the loss function is bounded in $[a, b]$, one can simply replace the $\sigma$ in~\eqref{equation::Bu et al IMI SGLD generalization bound} with $(b-a) / 2$.
However, as we see in \ref{subsection::examples}, the total variation achieves tighter bound than the KL-based one. Thus, the above result may be further improved. To this end, we invoke Corollary~\ref{corollary::expectational generalization bounds::TV} and obtain the following result \renjie{under the OOD setting}. The proof is postponed to Appendix~\ref{subsection::appendix::proof of theorem::SGLD using Pinsker and B-H}.

\begin{theorem}
    Let $W$ be the output of the SGLD algorithm as defined in~\eqref{equation::SGLD update rule}. Suppose the loss function is $L$-Lipschitz and bounded in $[a, b]$. We have
    {\small
    \begin{equation}
        \GeneralizationErrorAlgorithmic \leq (b-a)\TotalVariation{\nu}{\mu} + \frac{b-a}{n}\sum_{i=1}^n\Expectation{U_{[T]}}{\min\left\{
        \frac{1}{2}\sqrt{\sum_{t\in\mathcal{T}_i(U_{[T]})} \frac{\eta_t^2L^2}{\sigma_t^2}}
        ,\sqrt{1 - \exp\left(-\sum_{t\in\mathcal{T}_i(U_{[T]})} \frac{\eta_t^2L^2}{2\sigma_t^2}\right)}\right\}}.
        \label{equation::SGLD generalization bound using Pinsker and B-H inequality}
    \end{equation}}
    \label{theorem::SGLD generalization bound using Pinsker and B-H inequality}
\end{theorem}

The generalization bound~\eqref{equation::SGLD generalization bound using Pinsker and B-H inequality} consists of two parts. 
The first part, $(b-a)\TotalVariation{\nu}{\mu}$, captures the \renjie{distribution shift} and can further be bounded using $f$-divergence between $\nu$ and $\mu$. 
The remaining part of~\eqref{equation::SGLD generalization bound using Pinsker and B-H inequality} captures the in-distribution generalization gap of the SGLD algorithm. 
In particular, we notice that the first term inside the \renjie{$\min\{\cdot, \cdot\}$} coincides with the previous result~\eqref{equation::Bu et al IMI SGLD generalization bound}. 
This is also a direct consequence of the Pinsker inequality, as shown in Appendix~\ref{subsection::appendix::proof of theorem::SGLD using Pinsker and B-H}. 
In contrast, the second term inside the \renjie{$\min\{\cdot, \cdot\}$} follows from the Bretagnolle-Huber inequality \cite{bretagnolle1979estimation}, an exponential KL upper bound on the total variation. 

We compare the two terms inside the $\min\{\cdot,\cdot\}$ in Fig.~\ref{figure::Pinsker vs Bretagnolle-Huber}, where the x-axis is set to $\sum_{t\in\mathcal{T}_i}\frac{\eta_t^2L^2}{\sigma_t^2}$.
Since the parameters $\eta_t$ and $\sigma_t$ depend only on the SGLD algorithm, we keep them fixed and only consider the effect of the Lipschitz constant $L$ of the loss function. From Fig.~\ref{figure::Pinsker vs Bretagnolle-Huber}, we conclude that, if $L$ is small enough such that $\sum_{t\in\mathcal{T}_i}\dfrac{\eta_t^2L^2}{\sigma_t^2}<\delta$, 
then the Pinsker-type term is tighter. 
Here $\delta$ is the unique non-zero solution of the equation $\frac{1}{2}\sqrt{x} = \sqrt{1-e^{-x/2}}$, and we have
$\delta = 2W_0\left(-\dfrac{2}{e^2}\right) + 4\approx 3.187 $. 
Here, with a little abuse of notation, $W_0$ denotes the principal branch of the Lambert $W$ function.
On the other hand, if $L$ is sufficiently large such that $\sum_{t\in\mathcal{T}_i}\dfrac{\eta_t^2L^2}{\sigma_t^2}>\delta$, the Pinsker-type term begins to diverge and the newly added term helps establish a tighter generalization bound. 
\begin{figure*}[ht]
    \centering
    \includegraphics[width=0.5\textwidth]{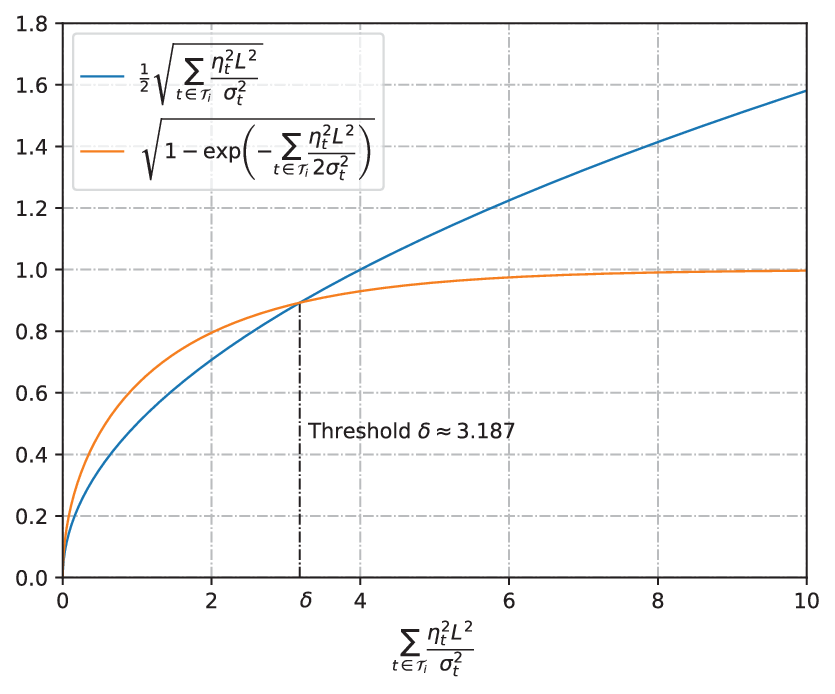}
    \caption{Comparison between the two terms in~\eqref{equation::SGLD generalization bound using Pinsker and B-H inequality}.}
    \label{figure::Pinsker vs Bretagnolle-Huber}
\end{figure*}

\subsection{SGLD using Without-Replacement Sampling Method}
\label{subsection::SGLD using Without Replacement Sampling}

In the subsequent two subsections, we turn to the CMI methods for the SGLD. 
The following theorem is a consequence of combining the $f$-ICIMI bound (Theorem~\ref{theorem::f-ICIMI}) with the squared Hellinger divergence. The proof is postponed to Appendix~\ref{subsection::appendix::proof of theorem::SGLD using chain rule}.

\begin{theorem}
    Let $W$ be the output of the SGLD algorithm as defined in~\eqref{equation::SGLD update rule}. Suppose the loss function is $L$-Lipschitz and bounded in $[a, b]$. We have
    \begin{equation}
        \GeneralizationErrorAlgorithmic \leq (b-a)\TotalVariation{\nu}{\mu} +  \frac{2(b-a)}{n}\sum_{t\in[T]}\sqrt{1 - \exp\left(-\frac{\eta_t^2L^2}{16\sigma_t^2}\right)}.
        \label{equation::f-ICIMI bounds on SGLD::chain rule}
    \end{equation}
    \label{theorem::f-ICIMI bounds on SGLD::chain rule}
\end{theorem}

To compare it with other in-distribution results, it suffices to consider the case $\nu=\mu$ so that the first TV term becomes 0.
By writing the summation $\sum_{t\in[T]}$ as $\sum_{i=1}^n\sum_{t\in\mathcal{T}_i}$, we see that the $\sum_{t\in\mathcal{T}_i}$ term is outside the square root as in~\eqref{equation::f-ICIMI bounds on SGLD::chain rule}, while the $\sum_{t\in\mathcal{T}_i}$ term is inside the square root as in~\eqref{equation::Bu et al IMI SGLD generalization bound} and~\eqref{equation::SGLD generalization bound using Pinsker and B-H inequality}. 
By the inequality $\sqrt{x+y}\leq\sqrt{x}+\sqrt{y}$, the generalization bound~\eqref{equation::f-ICIMI bounds on SGLD::chain rule} is larger than that of~\eqref{equation::Bu et al IMI SGLD generalization bound} and~\eqref{equation::SGLD generalization bound using Pinsker and B-H inequality} in general. However, Theorem~\ref{theorem::f-ICIMI bounds on SGLD::chain rule} 
shows its benefit if the SGLD is performed using the without-replacement sampling scheme, \ie, each data $Z_i$ is exactly applied once for updating $W$. 
Under this setting, $\mathcal{T}_i$ becomes a singleton and the summation term $\sum_{t\in\mathcal{T}_i}$ vanishes. 
Thus, under the without-replacement sampling scheme, the generalization bound in Theorem~\ref{theorem::f-ICIMI bounds on SGLD::chain rule} is strictly superior than the previous result~\eqref{equation::Bu et al IMI SGLD generalization bound}.
This is a direct consequence of the inequality $1-e^{-x}< x$ when $x>0$.

\subsection{SGLD with Asymptotic Assumption}
\label{subsection::SGLD with big O assumption}
As discussed in the previous subsection, the $\sum_{t\in\mathcal{T}_i}$ term of Theorem~\ref{theorem::f-ICIMI bounds on SGLD::chain rule} locates outside the square root, mostly causing a larger generalization bound if the replacement is allowed in sampling. 
Technically, this is a consequence of applying the chain rule of the squared Hellinger divergence in Theorem~\ref{theorem::f-ICIMI bounds on SGLD::chain rule}.  
Although the chain rule is a common technique applied in a series of KL-based SGLD generalization bounds, it does not perform well when other $f$-divergence is involved (at least for the squared Hellinger divergence).
In this subsection, we drop the chain rule and utilize the sub-additivity of the squared Hellinger divergence.
With the new technique at hand, we are able to move the $\sum_{t\in\mathcal{T}_i}$ term inside the square root and derive an asymptotically tighter generalization bound for \renjie{the} SGLD.

We begin with an intuitive observation. As the number of training data goes to infinity, the algorithm output is asymptotically independent of an individual sample.
Under the CMI model, this translates that the conditional probability of $R_i$ (identity of the $i$-th training sample), given the super-sample and some history outputs of the learning algorithm, should asymptotically converge to its marginal distribution $\frac{1}{2}\delta_1 + \frac{1}{2}\delta_{-1}$, \renjie{where $\delta_{1}$ is a Dirac delta distribution located at $1$}. The following assumption assumes there exists some function $g$ that captures the above convergence behavior. Recall that $n$ is the total number of training samples.

\begin{assumption}
    $\forall t\in \mathcal{T}_i,\ P_{R_i=1|W_{t-1},Z_i^\pm} = \dfrac{1}{2} + \mathcal{O}\left(g(n)\right)$, where $g\colon \mathbb{N}^+\to\RealNumber$ is some deterministic function satisfying $\lim_{n\to\infty}g(n) = 0$.
    \label{assumption::posterior probability of R_i is asymptotic 1/2}
\end{assumption}
\begin{remark}    
    One can imagine $g(n)$ as a polynomial $n^{-\gamma}$ for some $\gamma > 0$, but we do not add such restrictions. Additionally, one special case is the without-replacement sampling as discussed in the previous subsection. Since $R_i$ is exactly applied for updating $W$ at the $t$-th iteration, $W_{t-1}$ is independent of $R_i$, and thus $P_{R_i=1|W_{t-1},Z_i^\pm} = \frac{1}{2} = P_{R_i=-1|W_{t-1},Z_i^\pm}$ and $g(n)\equiv 0$. 
\end{remark}

\begin{theorem}
    Let $W$ be the output of the SGLD algorithm as defined in~\eqref{equation::SGLD update rule}. Suppose the loss function is $L$-Lipschitz and bounded in $[a, b]$. Under the Assumption~\ref{assumption::posterior probability of R_i is asymptotic 1/2}, we have
    \begin{equation}
         \GeneralizationErrorAlgorithmic \leq (b-a)\TotalVariation{\nu}{\mu} + \frac{2(b-a)}{n}\sum_{i=1}^n\Expectation{U_{[T]}}{
        \sqrt{\sum_{t\in\mathcal{T}_i(U_{[T]})}\left(1 - \exp\left(-\frac{\eta_t^2L^2}{16\sigma_t^2}\right) \right)}} + \mathcal{O}\left(g^2(n)\right).
        \label{equation::f-ICIMI bounds on SGLD::subadditivity}
    \end{equation}
    \label{theorem::f-ICIMI bounds on SGLD::subadditivity}
\end{theorem}
\begin{remark}
    Under the without-replacement sampling method, we have $g(n)\equiv 0$ and we exactly recover the generalization bound~\eqref{equation::f-ICIMI bounds on SGLD::chain rule}. Therefore, Theorem~\ref{theorem::f-ICIMI bounds on SGLD::chain rule} can be regarded as a special case of Theorem~\ref{theorem::f-ICIMI bounds on SGLD::subadditivity}.   
\end{remark}
Again, by the inequality $1-e^{-x}\leq x$, the generalization bound~\eqref{equation::f-ICIMI bounds on SGLD::subadditivity} is asymptotically tighter than the IMI bound~\eqref{equation::Bu et al IMI SGLD generalization bound} (with $\sigma=(b-a)/2$). Additionally, \cite{zhou2022individually} also derived an ICIMI-based bound for the SGLD algorithm. But it relies on a precise estimation of the posterior probability of $R_i$ (that is, $\alpha$ in our proof), which is computationally intractable. In the worst case, the ICIMI-bound can be twice larger than the IMI-bound. In contrast, the generalization bound in Theorem~\ref{theorem::f-ICIMI bounds on SGLD::subadditivity} is independent of the specific value of $\alpha$ but relies on its asymptotic behavior captured by some function $g(n)$, causing a $\mathcal{O}(g^2(n))$ error on the generalization bound. We conjecture that $g$ can be polynomial (\eg, $g(n) = 1 / n$) under certain sampling or training schemes, so that~\eqref{equation::f-ICIMI bounds on SGLD::subadditivity} can decay fast. We leave the study on function $g$ for future work.

The remaining part of this section focuses on proving the Theorem~\ref{theorem::f-ICIMI bounds on SGLD::subadditivity}. As mentioned before, we exploit the sub-additivity, rather than the chain rule, of the $f$-divergence. We start with its definition.

\begin{definition}[Sub-additivity of divergence \cite{ding2021gans}]
    Let $\mathcal{G}$ be a directed acyclic graph that defines a Bayes network. We say a divergence $D$ is sub-additive if, for all probability distributions $P$ and $Q$ over the Bayes network, we have
    \begin{equation}
        D(P||Q)\leq \sum_{v\in\mathcal{G}}D(P_{\{v\}\cup\Pi_v}||Q_{\{v\}\cup\Pi_v}),
    \end{equation}
    where $v$ denotes the vertex of $\mathcal{G}$ and $\Pi_v$ denotes the set of parents of $v$.
\end{definition}

Many $f$-divergences are sub-additive, but here we only need the sub-additivity of the squared Hellinger divergence.

\begin{lemma}[Theorem 2 of \cite{ding2021gans}]
    \label{lemma::some f divergence is sub-additive}
    The squared Hellinger divergence (\cf Table~\ref{table::comparison between f divergences}) is sub-additive on Bayes networks.
\end{lemma}

Combining the sub-additivity with SGLD yields the following useful lemma. The underlying reason why the sub-additivity can be applied to SGLD is that one can construct a Bayes network from the trajectory of SGLD. The proof is deferred to Appendix~\ref{subsection::appendix::proof of lemma::subadditivity of f-divergence}.

\begin{lemma}[Sub-additivity of Conditional $f$-Mutual Information]
    \label{lemma::subadditivity of f mutual information}
    Let $D_f$ be a sub-additive $f$-divergence, and $W_{[0:T]}\coloneqq(W_0, \ldots, W_T)$ be the trajectory generated by the SGLD algorithm under some fixed sample path $U_{[T]} = (u_1, \ldots, u_T)$. Again, let $R_i$ and $Z_i^\pm$ respectively be the data identifier and super-sample, as illustrated in the CMI model. We have
    \begin{equation}
        \fCMI{W_{[0:T]}}{R_i}{Z_i^\pm} \leq \sum_{t\in \mathcal{T}_i}\Expectation{W_{t-1},Z_i^\pm}{\fDivergence{\ProbabilityKernel{W_t, R_i}{W_{t-1},Z_i^\pm}}{\ProbabilityKernel{W_t}{W_{t-1},Z_i^\pm}P_{R_i}}},
        \label{equation::subadditivity of f mutual information}
    \end{equation}
    where $P_{R_i}$ denotes the uniform distribution over the set $\left\{-1, +1\right\}$.
\end{lemma}


Now we are ready to prove the main result, \renjie{\ie, Theorem \ref{theorem::f-ICIMI bounds on SGLD::subadditivity}}, of this subsection.
\begin{proof}[Proof of Theorem~\ref{theorem::f-ICIMI bounds on SGLD::subadditivity}]
    By Theorem~\ref{theorem::f-ICIMI}, we have
    \begin{align}
        \InDistributionGeneralizationGap
        &\leq \frac{b-a}{n}\sum_{i=1}^n\sqrt{\frac{2\fCMI[f]{W_T}{R_i}{Z_i^{\pm}}}{f''(1)}} \\
        & \leq \frac{b-a}{n}\sum_{i=1}^n\sqrt{\frac{2\fCMI[f]{W_{[0:T]}}{R_i}{Z_i^{\pm}}}{f''(1)}} \\
        &\leq \frac{b-a}{n}\sum_{i=1}^n\sqrt{\frac{2}{f''(1)}\sum_{t\in \mathcal{T}_i}\Expectation{W_{t-1},Z_i^\pm}{\fDivergence{\ProbabilityKernel{W_t, R_i}{W_{t-1},Z_i^\pm}}{\ProbabilityKernel{W_t}{W_{t-1},Z_i^\pm}P_{R_i}}}},
        \label{equation::upper bound of SGLD in terms of chain decomposition}
    \end{align}
    where the last inequality follows by Lemma~\ref{lemma::subadditivity of f mutual information}. 
    Let $\ProbabilityKernel{R_i}{W_{t-1},Z_i^\pm} = \alpha\delta_1 + (1-\alpha)\delta_{-1}$. By Assumption~\ref{assumption::posterior probability of R_i is asymptotic 1/2} we can write
    \begin{equation}
        \alpha(W_{t-1},Z_i^\pm) = \frac{1}{2}\left(1 + \epsilon(W_{t-1},Z_i^\pm)\right), \ \epsilon\in\mathcal{O}(g(n)).
        \label{equation::alpha = 1/2 + epsilon}
    \end{equation}
    It follows that $\ProbabilityKernel{W_t}{W_{t-1},Z_i^\pm,R_i}$ is Gaussian and $\ProbabilityKernel{W_t}{W_{t-1},Z_i^\pm}$ is a Gaussian mixture model (\cf~\eqref{equation::W_t given W_t-1 and Z_i and R_i is Gaussian} and~\eqref{equation::W_t given W_t-1 and Z_i is mixture Gaussian} in Appendix~\ref{subsection::appendix::proof of theorem::SGLD using chain rule}). We have
    \begin{align} 
        \ProbabilityKernel{W_t, R_i}{W_{t-1},Z_i^\pm} 
        &= \alpha \mathcal{N}^+\times \delta_1 + (1-\alpha)\mathcal{N}^{-} \times \delta_{-1}, 
        \label{equation::joint distribution of W_t and R_i conditioned on W_t-1 and Z_i}\\
        \ProbabilityKernel{W_t}{W_{t-1},Z_i^\pm}P_{R_i}
        &= \left(\alpha\mathcal{N}^+ + (1-\alpha)\mathcal{N}^{-}\right) \times \left(\frac{1}{2}\delta_1+\frac{1}{2}\delta_{-1}\right) \\
        &= \alpha\mathcal{N}^+\times \left(\frac{1}{2}\delta_1+\frac{1}{2}\delta_{-1}\right) + (1-\alpha)\mathcal{N}^{-} \times \left(\frac{1}{2}\delta_1+\frac{1}{2}\delta_{-1}\right).
        \label{equation::marginal distribution of W_t conditioned on W_t-1 and Z_i times P_R_i}
    \end{align}
    In what follows, we set $f(x)=\left(\sqrt{x}-1\right)^2$ and thus $D_f = H^2$, but we still use the $f$-notation for those expressions holding for arbitrary $f$-divergence. Plugging~\eqref{equation::joint distribution of W_t and R_i conditioned on W_t-1 and Z_i} and~\eqref{equation::marginal distribution of W_t conditioned on W_t-1 and Z_i times P_R_i} into~\eqref{equation::upper bound of SGLD in terms of chain decomposition} yields
    \begin{align}
        &\quad \ \Expectation{W_{t-1},Z_i^\pm}{\fDivergence{\ProbabilityKernel{W_t, R_i}{W_{t-1},Z_i^\pm}}{\ProbabilityKernel{W_t}{W_{t-1},Z_i^\pm}P_{R_i}}} \nonumber \\
        &\leq  \Expectation{W_{t-1},Z_i^\pm}{
        \alpha \fDivergence{\ProbabilityKernel{W_t, R_i}{W_{t-1},Z_i^\pm}}{\mathcal{N}^+\times \left(\frac{1}{2}\delta_1+\frac{1}{2}\delta_{-1}\right)} + (1-\alpha)\fDivergence{\ProbabilityKernel{W_t, R_i}{W_{t-1},Z_i^\pm}}{\mathcal{N}^-\times \left(\frac{1}{2}\delta_1+\frac{1}{2}\delta_{-1}\right)}}
        \label{equation::subadditivity::eq-1}\\
        &= \Expectation{W_{t-1},Z_i^\pm}{
        \frac{\alpha}{2}\left(f(2\alpha) + \Expectation{\mathcal{N}^+}{f\left(2(1-\alpha)\RadonNikodym{\mathcal{N}^-}{\mathcal{N}^+}\right)}\right) + 
        \frac{1-\alpha}{2}\left(f(2(1-\alpha)) + \Expectation{\mathcal{N}^-}{f\left(2\alpha\RadonNikodym{\mathcal{N}^+}{\mathcal{N}^-}\right)}\right)
        }\label{equation::subadditivity::eq-2}\\
        &= \Expectation{W_{t-1},Z_i^\pm}{2 - \sqrt{2}\left(
        \alpha^{3/2} + (1-\alpha)^{3/2} + \left(\alpha\sqrt{1-\alpha} + (1-\alpha)\sqrt{\alpha}\right)
        \int_{\mathcal{W}}\sqrt{\mathrm{d}\mathcal{N}^+\mathrm{d}\mathcal{N}^-}
        \right)}\label{equation::subadditivity::eq-3}\\
        &= \Expectation{W_{t-1},Z_i^\pm}{1 - \int_{\mathcal{W}}\sqrt{\mathrm{d}\mathcal{N}^+\mathrm{d}\mathcal{N}^-}} + \mathcal{O}\left(\epsilon^2\right)
        \label{equation::subadditivity::eq-4}\\
        &= \Expectation{W_{t-1},Z_i^\pm}{1 - \exp\left(-\frac{\eta_t^2\|\nabla_W\ell(W_{t-1},Z_i^+) - \nabla_W\ell(W_{t-1},Z_i^-)\|^2}{8\sigma_t^2}\right)} + \mathcal{O}\left(\epsilon^2\right)
        \label{equation::subadditivity::eq-5}\\
        &\leq 1 - \exp\left(-\frac{\eta_t^2L^2}{16\sigma_t^2}\right) + \mathcal{O}\left(\epsilon^2\right).
        \label{equation::upper bound of conditional f-mutual information with big O term}
    \end{align}
    In the above, inequality~\eqref{equation::subadditivity::eq-1} follows since the $f$-divergence is jointly (and thus separately) convex, and \eqref{equation::subadditivity::eq-2} follows by the definition of $f$-divergence. We derive \eqref{equation::subadditivity::eq-3} since $f(x)=(\sqrt{x}-1)^2$. To derive~\eqref{equation::subadditivity::eq-4}, we use 
    \begin{subequations}
        \begin{align}
            \alpha^{3/2} &= \left(\frac{1}{2}\right)^{3/2}\left(1 + \frac{3}{2}\epsilon\right) + \mathcal{O}\left(\epsilon^2\right), \\
            (1-\alpha)^{3/2} &= \left(\frac{1}{2}\right)^{3/2}\left(1 - \frac{3}{2}\epsilon\right) + \mathcal{O}\left(\epsilon^2\right), \\
            \alpha\sqrt{1-\alpha} &= \left(\frac{1}{2}\right)^{3/2}\left(1 + \frac{1}{2}\epsilon\right) + \mathcal{O}\left(\epsilon^2\right), \\
            (1-\alpha)\sqrt{\alpha} &= \left(\frac{1}{2}\right)^{3/2}\left(1 - \frac{1}{2}\epsilon\right) + \mathcal{O}\left(\epsilon^2\right),
        \end{align}
    \end{subequations}
    where these equations come from~\eqref{equation::alpha = 1/2 + epsilon} and the fact that $(1+\epsilon)^{1/2} = 1 + \epsilon/2 + \mathcal{O}(\epsilon^2)$.
    Expression~\eqref{equation::subadditivity::eq-5} follows by Lemma~\ref{lemma::Squared Hellinger Divergence between two Gaussian Distribution} in Appendix \ref{subsection::appendix::proof of theorem::SGLD using chain rule} since $1 - \int_{\mathcal{W}}\sqrt{\mathrm{d}\mathcal{N}^+\mathrm{d}\mathcal{N}^-} = \frac{1}{2}\SquaredHellingerDivergence{\mathcal{N}^+}{\mathcal{N}^-}$, and the last inequality~\eqref{equation::upper bound of conditional f-mutual information with big O term} follows by~\eqref{equation::upper bound of f mutual information between W_t and R_i conditioned on W_t-1 and Z_i::Chain rule of H^2} \renjie{in Appendix \ref{subsection::appendix::proof of theorem::SGLD using chain rule}}. 
    
    Finally, inserting~\eqref{equation::upper bound of conditional f-mutual information with big O term} into~\eqref{equation::upper bound of SGLD in terms of chain decomposition} and taking expectation \wrt the sample path $U_{[T]}$, the desired result~\eqref{equation::f-ICIMI bounds on SGLD::subadditivity} follows since $\epsilon^2\in\mathcal{O}\left(g^2(n)\right)$.
\end{proof}

\section{Conclusions}
\label{section::Conclusions}
This paper proposes an information-theoretic framework for \renjie{analyzing} OOD generalization, which seamlessly interpolates between IPM and $f$-divergence, reproduces known results as well as yields new generalization bounds. 
Next, the framework is integrated with the CMI methods, resulting in a class of $f$-ICIMI bounds, among which includes the state-of-the-art ICIMI bound. 
These results separately give a natural extension and improvement to the IMI methods as well as the CMI methods.
Lastly, we demonstrate their applications by proving generalization bounds for the SGLD algorithm, under the premise that the loss function is both bounded and Lipschitz. 
Although these results are derived under the OOD case, they also demonstrate the benefits for the in-distribution case.
In particular, the study gives special attention to two specific scenarios: SGLD employing without-replacement sampling and SGLD with asymptotic posterior probability. 
Compared with the traditional chain-rule based methods, our approaches also shed light on the alternative way to study the generalization of SGLD algorithms through the subadditivity of $f$-divergence.

\appendices
\section{Proof of Section~\ref{section::A general theorem}}
\subsection{Proof of Proposition~\ref{proposition::fundamental inequality}}
\label{subsection::appendix::proof of the fundamental inequality}

We provide an alternative proof of Proposition~\ref{proposition::fundamental inequality}, demonstrating its relationship with $(f,\Gamma)$-divergence \cite{birrell2022f}. We start with its definition. 
\begin{definition}[$(f, \Gamma)$-Divergence \cite{birrell2022f}] 
    Let $\mathcal{X}$ be a probability space. Suppose $P, Q\in\mathcal{P}(\mathcal{X})$ and $\Gamma\subseteq\mathcal{M}(\mathcal{X})$, $f$ be the convex function that induces the $f$-divergence. The $(f,\Gamma)$-divergence between distribution $P$ and $Q$ is defined by
    \begin{equation}
        \fGammaDivergence{P}{Q} \coloneq \sup_{g\in \Gamma}\bigl\{\Expectation{P}{g} - \GeneralizedCGF{f}{Q}{g}\bigr\}.
    \end{equation}
    \label{definition::f-Gamma-divergence}
\end{definition}
The $(f,\Gamma)$-divergence admits an upper bound, which interpolates between $\Gamma$-IPM and $f$-divergence. 
\begin{lemma}(\!\cite[Theorem 8]{birrell2022f})
    \begin{equation}
        \fGammaDivergence{P}{Q} \leq \inf_{\eta\in\mathcal{P}(\mathcal{X})} 
        \left\{\GammaIPM{P}{\eta} + \fDivergence{\eta}{Q}\right\}.
        \label{equation::upper bounds of f-Gamma-divergence}
    \end{equation}
    \label{lemma::upper bounds of f-Gamma-divergence}
\end{lemma}
Now we are ready to prove Proposition~\ref{proposition::fundamental inequality}.
\begin{proof}[Proof of Proposition~\ref{proposition::fundamental inequality} using $(f,\Gamma)$-Divergence]
    \begin{align}
        \GeneralizationErrorAlgorithmic 
        &= \frac{1}{n} \sum_{i=1}^n \Expectation{P_i}{\RecenteredLossFunction{W}{Z_i}} 
        \label{eq:alternative proof:prop1-eq1}\\
        &= \frac{1}{n} \sum_{i=1}^n \frac{1}{t_i}\Expectation{P_i}{t_i\RecenteredLossFunction{W}{Z_i}} \\
        &\leq \frac{1}{n} \sum_{i=1}^n \frac{1}{t_i}\left(D_{f}^{t_i\Bar{\Gamma}}\left(P_i || Q \right) + \GeneralizedCGF{f}{Q}{t_i\RecenteredLossFunction{W}{Z_i}}\right) 
        \label{eq:alternative proof:prop1-eq2} \\
        & \leq \frac{1}{n} \sum_{i=1}^n \frac{1}{t_i}\inf_{\eta_i\in\mathcal{P}\left(\mathcal{W}\times\mathcal{Z}\right)}
        \Bigl\{W^{t_i\Bar{\Gamma}}\left(P_i,\eta_i\right)+ \fDivergence{\eta_i}{Q} + \GeneralizedCGF{f}{Q}{t_i\RecenteredLossFunction{W}{Z_i}}\Bigr\} 
        \label{eq:alternative proof:prop1-eq3} \\
        &= \text{RHS of}~\eqref{equation::fundamental inequality}. \nonumber
    \end{align}
    Here equality~\eqref{eq:alternative proof:prop1-eq1} follows by~\eqref{eq:prop1-eq4}, inequality~\eqref{eq:alternative proof:prop1-eq2} follows by Definition~\ref{definition::f-Gamma-divergence} and the condition $t_i\Bar{\ell}\in t_i\Bar{\Gamma}$, inequality~\eqref{eq:alternative proof:prop1-eq3} follows by Lemma~\ref{lemma::upper bounds of f-Gamma-divergence}, and the last equality follows by the fact that $\dfrac{1}{t}W^{t\Bar{\Gamma}}\left(P_i,\eta_i\right) = \GammaBarIPM{P_i}{\eta_i}$, for all $t\in\RealNumberNonnegative$.
\end{proof}

\subsection{Tightness of the Proposition~\ref{proposition::fundamental inequality}}
\label{subsection::appendix::tightness of the fundamental inequality}
The following proposition says that the equality in Proposition~\ref{proposition::fundamental inequality} can be achieved under certain conditions. 
\begin{proposition} The upper bound in Proposition~\ref{proposition::fundamental inequality} achieves equality if the following two conditions hold simultaneously.
    \begin{enumerate}
        \item $\Bar{\Gamma}$ is a singleton, \ie, $\Bar{\ell}$ is the only element of $\Bar{\Gamma}$.
        \label{enumeration::assumption-1}
        \item For each $i=1,\ldots,n$, the distribution $\eta_i$ and the parameter $t_i$ are related through
        \label{enumeration::assumption-2}
        \begin{equation}
            \RadonNikodym{\eta_i}{Q} = (f^*)'\left(t_i\RecenteredLossFunction{w}{z} - \lambda_i\right),
            \label{equation::the optimal eta}
        \end{equation}
        where $\lambda_i \in \RealNumber$ makes~\eqref{equation::the optimal eta} a probability density:
        \begin{equation}
            \Expectation{Q}{(f^*)'\left(t_i\RecenteredLossFunction{W}{Z} - \lambda_i\right)} = 1.
        \end{equation}
    \end{enumerate}
    \label{proposition::tightness of the fundamental inequality}
\end{proposition}
\begin{remark}
    Under the case of KL-divergence (see Remark~\ref{remark::cumulant generating function}), we have $(f^*)'(x) = e^x$ and thus $\lambda_i = \log\Expectation{Q}{e^{t_i\Bar{\ell}(W,Z)}}$. Therefore, the optimal $\eta_i$ has the form of
    \begin{equation}
        \RadonNikodym{\eta_i}{Q}(w,z) = \frac{e^{t_i\Bar{\ell}(w,z)}}{\Expectation{Q}{e^{t_i\Bar{\ell}(W,Z)}}} 
        = \frac{e^{-t_i\ell(w,z)}}{\Expectation{Q}{e^{-t_i\ell(W,Z)}}}.
    \end{equation}
    This means that the optimal $\eta_i$ is achieved exactly at the Gibbs posterior distribution, with $t_i$ acting as the inverse temperature. 
\end{remark}
\begin{proof}[Proof of Proposition~\ref{proposition::tightness of the fundamental inequality}]
    By assumption~\ref{enumeration::assumption-1}, we have $\GammaBarIPM{P_i}{\eta_i} = \Expectation{P_i}{\Bar{\ell}} - \Expectation{\eta_i}{\Bar{\ell}}$, and thus Proposition~\ref{proposition::fundamental inequality} becomes
    \begin{equation}
        \GeneralizationErrorAlgorithmic\leq \frac{1}{n}\sum_{i=1}^{n}\Bigl(\Expectation{P_i}{\Bar{\ell}} - \Expectation{\eta_i}{\Bar{\ell}}  \\
        + \frac{1}{t_i}\fDivergence{\eta_i}{Q} + \frac{1}{t_i}\GeneralizedCGF{f}{Q}{t_i\RecenteredLossFunction{W}{Z}}\Bigr).
    \end{equation}
    As a consequence, it suffices to prove
    \begin{equation}
        \Expectation{\eta}{g} = \frac{1}{t}\fDivergence{\eta}{Q} + \frac{1}{t}\GeneralizedCGF{f}{Q}{tg},
        \label{equation::appendix::equivalent equality}
    \end{equation}
    under the conditions that
    \begin{subequations}
    \begin{align}
        & \RadonNikodym{\eta}{Q} = (f^*)'\left(t(g-\lambda)\right),
        \label{enumeration::equivalent condition-1}\\
        & \Expectation{Q}{(f^*)'\left(t(g-\lambda)\right)} = 1,
        \label{enumeration::equivalent condition-2}
    \end{align}
    \end{subequations}
    where $\eta, Q\in\mathcal{P}(\mathcal{X})$, $g\in\mathcal{M}(\mathcal{X})$, and $t\in\RealNumberNonnegative$.
    If it is the case, then Proposition~\ref{proposition::tightness of the fundamental inequality} follows by setting $\mathcal{X} = \mathcal{W}\times\mathcal{Z}$, $\eta=\eta_i$, $t=t_i$, $g=\Bar{\ell}$, and $\lambda=\frac{1}{t_i}\lambda_i$. To see~\eqref{equation::appendix::equivalent equality} holds, we need the following lemma.
    \begin{lemma}(\cite[Lemma 48]{birrell2022f})
        \begin{equation}
            f\bigl((f^*)'(y)\bigr) = y(f^*)'(y) - f^*(y).
        \end{equation}
        \label{lemma::lemma for prove the tightness}
    \end{lemma}
    Then the subsequent argument is very similar to that of \cite[Theorem 82]{birrell2022f}. We have
    \begin{align}
        &\quad \  \sup_{P\in\mathcal{P}(\mathcal{X})}\Bigl\{\Expectation{P}{g} - \frac{1}{t}\fDivergence{P}{Q}\Bigr\} \\
        & \geq \lambda + \Expectation{\eta}{g-\lambda} - \frac{1}{t}\fDivergence{\eta}{Q} \\
        & = \lambda + \Expectation{Q}{(f^*)'(t(g-\lambda))(g-\lambda)} - \frac{1}{t}\fDivergence{\eta}{Q} 
        \label{equation::appendix::proof of the tightness::eq-1}\\
        & = \frac{1}{t}\bigl( t\lambda + \Expectation{Q}{f^*(t(g-\lambda))}\bigr) 
        \label{equation::appendix::proof of the tightness::eq-2}\\
        & \geq \frac{1}{t}\GeneralizedCGF{f}{Q}{tg} 
        \label{equation::appendix::proof of the tightness::eq-3}\\
        & = \sup_{P\in\mathcal{P}(\mathcal{X})}\Bigl\{\Expectation{P}{g} - \frac{1}{t}\fDivergence{P}{Q}\Bigr\}.
        \label{equation::appendix::proof of the tightness::eq-4}
    \end{align}
    In the above, equality~\eqref{equation::appendix::proof of the tightness::eq-1} follows by~\eqref{enumeration::equivalent condition-1}, equality~\eqref{equation::appendix::proof of the tightness::eq-2} follows by Lemma~\ref{lemma::lemma for prove the tightness}, inequality~\eqref{equation::appendix::proof of the tightness::eq-3} follows by Definition~\ref{definition::generalized CGF}, and equality~\eqref{equation::appendix::proof of the tightness::eq-4} follows by Lemma~\ref{lemma::variational representation of f divergence} and equality~\eqref{equation::appendix::legendre dual of tF}.
    Therefore, all the inequalities above achieve equality. This proves~\eqref{equation::appendix::equivalent equality}.
\end{proof}
\section{Proofs in Section~\ref{section::corollaries}}
\label{section::proofs of the special cases}

\subsection{Proof of Corollary~\ref{corollary::expectational generalization bounds::Wasserstein}}
\label{subsection::appendix::proof of Wasserstein bounds}
\begin{proof}
    By Corollary~\ref{corollary::the general theorem::Gamma}, we have
    \begin{align}
        \GeneralizationErrorAlgorithmic 
        & \leq \frac{1}{n}\sum_{i=1}^n \sup_{g\in\Gamma} \bigl\{\Expectation{P_i}{g} - \Expectation{Q}{g}\bigr\} \\
        & = \frac{1}{n}\sum_{i=1}^n \sup_{g\in\Gamma}\bigl\{\Expectation{P_i}{g}
            - \Expectation{P_W\otimes \nu}{g} + \Expectation{P_W\otimes \nu}{g} 
            - \Expectation{Q}{g}\bigr\}\\
        & \leq \frac{1}{n}\sum_{i=1}^n \sup_{g\in\Gamma}\Bigl\{
        \Expectation{\nu}{\Expectation{\ProbabilityKernel{W}{Z_i}}{g} - \Expectation{P_W}{g}}+ \Expectation{P_W}{\Expectation{\nu}{g}  - \Expectation{\mu}{g}}\Bigr\}
        \label{equation::appendix::wasserstein::eq-1}\\
        & \leq \frac{1}{n}\sum_{i=1}^n \Expectation{\nu}{L_W W_1(\ProbabilityKernel{W}{Z_i}, P_W)} + L_Z W_1(\nu,\mu).
        \label{equation::appendix::wasserstein::eq-2}
    \end{align}
In the above, inequality~\eqref{equation::appendix::wasserstein::eq-1} follows by the tower property of conditional expectation, and inequality~\eqref{equation::appendix::wasserstein::eq-2} follows by the Kantorovich-Rubinstein duality~\eqref{equation::Kantorovich-Rubinstein Duality}. 
\end{proof}

\subsection{Proof of Corollary~\ref{corollary::expectational generalization bounds::TV}}
\label{subsection::appendix::proof of TV bounds}
\begin{proof}
    By assumption we have $\ell\in\Gamma$ and thus 
    \begin{align}
        \GeneralizationErrorAlgorithmic 
        &\leq \frac{1}{n}\sum_{i=1}^n \GammaIPM{P_i}{Q} \label{equation::appendix::TV::eq-1} \\
        & = \frac{1}{n}\sum_{i=1}^n W^{\Gamma-(b-a)/2}(P_i,Q) \label{equation::appendix::TV::eq-2} \\
        & = \frac{b-a}{n}\sum_{i=1}^n \TotalVariation{P_i}{Q} \label{equation::appendix::TV::eq-3}.
    \end{align}
    In the above, inequality~\eqref{equation::appendix::TV::eq-1} follows by Corollary~\ref{corollary::the general theorem::Gamma}, equality~\eqref{equation::appendix::TV::eq-2} follows by the translation invariance of IPM, and equality~\eqref{equation::appendix::TV::eq-3} follows by the variational representation of total variation:
    \begin{equation}
        \TotalVariation{P}{Q} = \sup_{\|g\|_{\infty}\leq \frac{1}{2}}\bigl\{\Expectation{P}{g} - \Expectation{Q}{g}\bigr\}.
    \end{equation}
    Thus we proved~\eqref{equation::expected generalization bounds::TV::eq-1}. Then~\eqref{equation::expectational generalization bounds::TV::eq-2} follows by the chain rule of total variation. The general form of the chain rule of total variation is given by
    \begin{equation}
        \TotalVariation{P_{X^m}}{Q_{X^m}} \leq \sum_{i=1}^m \Expectation{P_{X^{i-1}}}{\TotalVariation{P_{X_i|X^{i-1}}}{Q_{X_i|X^{i-1}}}}.
    \end{equation}
\end{proof}

\subsection{Proof of Corollaries~\ref{corollary::expectational generalization bounds::KL sub gaussian} and~\ref{corollary::expectational generalization bounds::KL sub gamma}}
\label{subsection::appendix::proof of KL bounds}
\begin{proof}
    It suffices to prove Corollary~\ref{corollary::expectational generalization bounds::KL sub gaussian} and then Corollary~\ref{corollary::expectational generalization bounds::KL sub gamma} follows by a similar argument. By Theorem~\ref{theorem::the general theorem}, we have
    \begin{align}
        \GeneralizationErrorAlgorithmic 
        &\leq \frac{1}{n}\sum_{i=1}^n\sqrt{2\sigma^2\KLDivergence{P_i}{Q}} \\
        &= \frac{1}{n}\sum_{i=1}^n\sqrt{2\sigma^2\bigl(\ConditionalKLDivergence{\ProbabilityKernel{W}{Z_i}}{Q_W}{\nu} + \KLDivergence{\nu}{\mu}\bigr)},
    \end{align}
    where the equality follows from the chain rule of KL divergence. Taking infimum over $Q_W$ yields~\eqref{equation::expectational generalization bounds::KL sub gaussian}, which is due to the following lemma.
    \begin{lemma}[Theorem 4.1 in~\cite{polyanskiy2022information}]
        Suppose $(W, Z)$ is a pair of random variables with marginal distribution $P_W$ and let $Q_W$ be an arbitrary distribution of $W$. If $\KLDivergence{P_W}{Q_W}<\infty$, then 
        \begin{equation}
            I(W;Z) = \ConditionalKLDivergence{\ProbabilityKernel{W}{Z}}{Q_W}{Z} - \KLDivergence{P_W}{Q_W}.
        \end{equation}
    \end{lemma}
    By the non-negativity of KL divergence, the infimum is achieved at $Q_W = P_W$ and thus $I(W; Z_i) = \ConditionalKLDivergence{\ProbabilityKernel{W}{Z_i}}{P_W}{\nu}$.
\end{proof}

\subsection{Proof of Corollary~\ref{corollary::expectational generalization bounds::Chi Square}}
\label{subsection::appendix::proof of chi square bounds}
\begin{proof}
    A direct calculation shows $f^*(y) = \frac{1}{4}y^2 + y$ for $f(x) = (x-1)^2$, and thus
    $\GeneralizedCGF{f}{\mu}{t\Bar{\ell}(w,Z)} = \frac{1}{4}\mathrm{Var}_{\mu}\LossFunction{w}{Z}t^2$. 
    Therefore, we can choose $\psi(t) = \frac{1}{4}\sigma^2t^2$ and thus $(\psi^*)^{-1}(y) = \sqrt{\sigma^2y}$. Applying Theorem~\ref{theorem::the general theorem} yields~\eqref{equation::chi square generalization bounds}.
\end{proof}



\subsection{Proof of Corollary~\ref{corollary::pp generalization bounds}}
\label{subsection::appendix::proof of PP generaliation bounds}
\begin{proof}
    Since $\Bar{\Gamma}=\{\Bar{\ell}\}$, we have
    \begin{align}
        \GammaBarIPM{P_i}{\eta_i} &= \Expectation{P_i}{\Bar{\ell}} - \Expectation{\eta_i}{\Bar{\ell}} \\
        &= \Expectation{P_W\otimes\nu}{\ell} - \Expectation{\ProbabilityKernel{W}{Z_i}\otimes\nu}{\ell}.
        \label{equation::appendix::PP generalization bounds::eq-1}
    \end{align}
    Inserting~\eqref{equation::appendix::PP generalization bounds::eq-1} into Theorem~\ref{theorem::the general theorem} and rearranging terms yields
    \begin{align}
        \Expectation{P_W\otimes\mu}{\ell} - \Expectation{P_W\otimes\nu}{\ell} 
        &\leq (\psi^*)^{-1}(\fDivergence{P_W\otimes\nu}{P_W\otimes\mu}) \\
        & = (\psi^*)^{-1}(\fDivergence{\nu}{\mu}).
    \end{align}
\end{proof}


\subsection{Derivation of Estimating the Gaussian and Bernoulli Means}
\label{subsection::appendix::Gaussian and Bernoulli means}
To calculate the generalization bounds, we need the distribution $P_i$ and $Q$. All the following results are given in the general $d$-dimensional case, where we let the training distribution be $\mathcal{N}(\mathbf{m},\sigma^2\mathbf{I}_d)$ and the testing distribution be $\mathcal{N}(\mathbf{m}',(\sigma')^2\mathbf{I}_d)$. 

We can check that both $P_i$ and $Q$ are joint Gaussian. Write the random vector as $\displaystyle[\mathbf{Z}^{\mathrm{T}}, \mathbf{W}^{\mathrm{T}}]^{\mathrm{T}}$, then $P_i$ and $Q$ are given by
\begin{align}
    & P_i = \mathcal{N}\left(\left[
    \begin{array}{cc}
       \mathbf{m}    \\[4pt]
       \mathbf{m}   
    \end{array}
    \right], \left[
    \begin{array}{cc}
         \sigma^2\mathbf{I}_d, &  \frac{1}{n}\sigma^2\mathbf{I}_d\\[4pt]
         \frac{1}{n}\sigma^2\mathbf{I}_d, & \frac{1}{n}\sigma^2\mathbf{I}_d
    \end{array}
    \right]\right), \\
    & Q = \mathcal{N}\left(\left[
    \begin{array}{cc}
       \mathbf{m}'    \\ [2pt]
       \mathbf{m}   
    \end{array}
    \right], \left[
    \begin{array}{cc}
         (\sigma')^2\mathbf{I}_d, &  \mathbf{0}\\[2pt]
         \mathbf{0}, & \frac{1}{n}\sigma^2\mathbf{I}_d
    \end{array}
    \right]\right).
\end{align}
The KL divergence between $P_i$ and $Q$ is given by
\begin{equation}
    \KLDivergence{P_i}{Q} = \log\frac{\det \mathbf{\Sigma}_{P_i}}{\det \mathbf{\Sigma}_{Q}} - 2d + \mathrm{Tr}(\mathbf{\Sigma}_{P_i}\mathbf{\Sigma}_Q^{-1}) \\
    + \exp\bigl((\mathbf{m}_{P_i} - \mathbf{m}_Q)^{\mathrm{T}}\mathbf{\Sigma}_Q^{-1}(\mathbf{m}_{P_i} - \mathbf{m}_Q)\bigr),
\end{equation}
where $\mathbf{m}_{P_i}$ (\emph{resp.}, $\mathbf{m}_Q$) denotes the mean vector of $P_i$ (\emph{resp.}, $Q$), and $\mathbf{\Sigma}_{P_i}$ (\emph{resp.}, $\mathbf{\Sigma}_Q$) denotes the covariance matrix of $P_i$ (\emph{resp.}, $Q$).
The $\chi^2$ divergence between $P_i$ and $Q$ is given by
\begin{equation}
    \ChiSquareDivergence{P_i}{Q} = \frac{\det\mathbf{\Sigma}_{Q}}{\sqrt{\det\mathbf{\Sigma}_{P_i}}\sqrt{\det\bigl(2\mathbf{\Sigma}_Q - \mathbf{\Sigma}_{P_i}\bigr)}}\cdot \\
    \exp\Bigl((\mathbf{m}_{P_i} - \mathbf{m}_Q)^{\mathrm{T}}\bigl(2\mathbf{\Sigma}_Q - \mathbf{\Sigma}_P\bigr)^{-1}(\mathbf{m}_{P_i} - \mathbf{m}_Q)\Bigr) - 1.
\end{equation}
Finally, the true generalization gap is given by
\begin{equation}
    \GeneralizationErrorAlgorithmic = \left((\sigma')^2 - \sigma^2\right)d + \frac{2\sigma^2d}{n} + \|\mathbf{m}-\mathbf{m}'\|_2^2.
\end{equation}

As for the example of estimating Bernoulli examples, a direct calculation shows 
\begin{equation}
    P_i\left(
    \begin{aligned}
        &Z_i=1 \\
        &W=\frac{k}{n}
    \end{aligned} 
    \right) = \left\{
    \begin{aligned}
        &\binom{n-1}{k-1}p^{k}(1-p)^{n-k-1}, 1\leq k\leq n, \\
        &0, k = 0,        
    \end{aligned}
    \right.
\end{equation}
\begin{equation}
    P_i\left(
    \begin{aligned}
        &Z_i=0 \\
        &W=\frac{k}{n}
    \end{aligned} 
    \right) = \left\{
    \begin{aligned}
        &\binom{n-1}{k}p^{k}(1-p)^{n-k}, 0\leq k\leq n-1, \\
        &0, k = n.    
    \end{aligned}
    \right.
\end{equation}
The distribution $Q$ is the product of $\mathrm{Bern}(p')$ and the binomial distribution with parameter $(n,p)$. 
Then the $f$-divergence can be directly calculated by definition. Finally, the true generalization gap is given by 
\begin{equation}
    \GeneralizationErrorAlgorithmic = 
    2\sum_{k=1}^n\binom{n-1}{k-1}p^k(1-p)^{n-1}\frac{k}{n} \\
    + (1-2p)p'-p.
\end{equation}

\section{Proofs in Section~\ref{section::Applications on the SGLD Algorithm}}
\label{section::appendix::proofs of the theorems about SGLD}

\subsection{Proof of Theorem~\ref{theorem::SGLD generalization bound using Pinsker and B-H inequality}}
\label{subsection::appendix::proof of theorem::SGLD using Pinsker and B-H}

The following lemma will be useful to the proof, which is implicitly included in~\cite{bu2020tightening}. 
    \begin{lemma}
        Suppose the loss function is $L$-Lipschitz, then, 
        \begin{equation}
            I\left(W_t;Z_i|W_{t-1}\right)\leq \frac{\eta_t^2L^2}{2\sigma_t^2}.
        \end{equation}
        \label{lemma::upper bound of I(W_t;Z_i|W_{t-1})}
    \end{lemma}
\begin{proof}[Proof of Theorem~\ref{theorem::SGLD generalization bound using Pinsker and B-H inequality}]
    Invoking Corollary~\ref{corollary::expectational generalization bounds::TV} and it suffices to bound the in-distribution term $\Expectation{\nu}{\TotalVariation{P_{W|Z_i}}{P_{W}}}$. Let the sample path $U_{[T]}$ be fixed, we have
    \begin{align}
        \Expectation{\nu}{\TotalVariation{P_{W|Z_i}}{P_{W}}}
        &\leq \sqrt{\Expectation{\nu}{\mathrm{TV}^2\left(P_{W|Z_i}, P_{W}\right)}} 
        \label{equation::TV-KL-SGLD::eq-1-1}\\
        &\leq \sqrt{\Expectation{\nu}{1 - \exp\left(-\KLDivergence{P_{W|Z_i}}{P_{W}}\right)}}
        \label{equation::TV-KL-SGLD::eq-1-2}\\
        &\leq \sqrt{1 - \exp\left(-\Expectation{\nu}{\KLDivergence{P_{W|Z_i}}{P_{W}}}\right)}
        \label{equation::TV-KL-SGLD::eq-1-3}\\
        &= \sqrt{1 - \exp\left(-I\left(W;Z_i\right)\right)}
        \label{equation::TV-KL-SGLD::eq-1-4}\\
        &\leq \sqrt{1 - \exp\left(-I\left(W_{[0:T]};Z_i\right)\right)}
        \label{equation::TV-KL-SGLD::eq-1-5}\\
        & = \sqrt{1 - \exp\left(-\sum_{t=1}^{T}I\left(W_t;Z_i|W_{t-1}\right)\right)}
        \label{equation::TV-KL-SGLD::eq-1-6}\\
        &\leq \sqrt{1 - \exp\left(-\sum_{t\in\mathcal{T}_i} \frac{\eta_t^2L^2}{2\sigma_t^2}\right)}.
        \label{equation::TV-KL-SGLD::eq-1-7}
    \end{align}
    In the above, inequalities~\eqref{equation::TV-KL-SGLD::eq-1-1} and~\eqref{equation::TV-KL-SGLD::eq-1-3} follow by Jensen's inequality, combined with the fact that the functions $\sqrt{x}$ and $1 - e^{-x}$ are both concave. Inequality~\eqref{equation::TV-KL-SGLD::eq-1-2} follows by the Bretagnolle-Huber inequality. To achieve inequality~\eqref{equation::TV-KL-SGLD::eq-1-5}, we use the inequality $I(W;Z_i)\leq I(W_{[0:T]};Z_i)$ and the fact that function $1-e^{-x}$ is monotonously increasing. 
    Equality~\eqref{equation::TV-KL-SGLD::eq-1-6} follows by the chain rule of mutual information, and the last inequality follows by Lemma~\ref{lemma::upper bound of I(W_t;Z_i|W_{t-1})} and the fact that $W_t$ is independent of $Z_i$ given $W_{t-1}$ if $t\notin \mathcal{T}_i$.

    On the other hand, we use Pinsker's inequality to yield
    \begin{align}
        \Expectation{\nu}{\TotalVariation{P_{W|Z_i}}{P_{W}}}
        &\leq \Expectation{\nu}{\sqrt{\frac{1}{2}\KLDivergence{P_{W|Z_i}}{P_W}}}
        \label{equation::KL-TV-SGLD::eq-2-1}\\
        &\leq \sqrt{\frac{I(W;Z_i)}{2}} 
        \label{equation::KL-TV-SGLD::eq-2-2}\\
        &\leq \frac{1}{2}\sqrt{\sum_{t\in\mathcal{T}_i}\frac{\eta_t^2L^2}{\sigma_t^2}},
        \label{equation::KL-TV-SGLD::eq-2-3}
    \end{align}
    where the last inequality follows by a similar argument from~\eqref{equation::TV-KL-SGLD::eq-1-4} to~\eqref{equation::TV-KL-SGLD::eq-1-7}. Therefore, the 
    $\Expectation{\nu}{\TotalVariation{P_{W|Z_i}}{P_{W}}}$ term is dominated by the minimum of~\eqref{equation::TV-KL-SGLD::eq-1-7} and~\eqref{equation::KL-TV-SGLD::eq-2-3}. Plug this result into Corollary~\ref{corollary::expectational generalization bounds::TV}, and the proof is completed by taking expectation \wrt the sample path $U_{[T]}$.
\end{proof}

\subsection{Proof of Theorem~\ref{theorem::f-ICIMI bounds on SGLD::chain rule}}
\label{subsection::appendix::proof of theorem::SGLD using chain rule}

For the proof of Theorem~\ref{theorem::f-ICIMI bounds on SGLD::chain rule}, the following lemmas will be useful.
\begin{lemma}[Chain Rule of Hellinger Distance] Let $P$ and $Q$ be two probability distributions on the random vector $X_{[n]} = (X_1,\ldots,X_n)$. Then,
    \label{lemma::Chain rule of squared Hellinger divergence}
    \begin{equation}
        \HellingerDistance{P_{X_{[n]}}}{Q_{X_{[n]}}}\leq \sum_{i=1}^n\sqrt{\Expectation{P_{X_{[i-1]}}}{\SquaredHellingerDivergence{P_{X_i|X_{[i-1]}}}{Q_{X_i|X_{[i-1]}}}}}.
    \end{equation}
\end{lemma}

\begin{lemma}[Squared Hellinger Divergence between two Gaussian Distribution]
    \label{lemma::Squared Hellinger Divergence between two Gaussian Distribution}
    Let $P = \mathcal{N}(\mathbf{\mu}_1, \mathbf{\Sigma}_1)$ and $Q = \mathcal{N}(\mathbf{\mu}_2, \mathbf{\Sigma}_2)$ be two multivariate Gaussian distributions. Then the squared Hellinger divergence of $P$ and $Q$ can be written as
    \begin{equation}
        \SquaredHellingerDivergence{P}{Q} = 2 - 
        2\frac{\det(\mathbf{\Sigma}_1)^{1/4}\det(\mathbf{\Sigma}_2)^{1/4}}{\det\left(\frac{\mathbf{\Sigma}_1+\mathbf{\Sigma}_2}{2}\right)^{1/2}}
        \exp\left(-\frac{1}{8}(\mathbf{\mu}_1 - \mathbf{\mu}_2)^{\mathrm{T}}\left(\dfrac{\mathbf{\Sigma}_1 + \mathbf{\Sigma}_2}{2}\right)^{-1}(\mathbf{\mu}_1 - \mathbf{\mu}_2)\right).
    \end{equation}
\end{lemma}
\begin{proof}[Proof Sketch]
    It follows from the Gaussian integral $\displaystyle\int_{\RealNumber^n}\exp\left(-\frac{1}{2}\mathbf{x}^{\mathrm
    T}\mathbf{A}\mathbf{x} + \mathbf{b}^{\mathrm{T}}\mathbf{x}\right)\mathrm{d}\mathbf{x} = \sqrt{\frac{(2\pi)^n}{\det(\mathbf{A})}}\exp\left(\frac{1}{2}\mathbf{b}^{\mathrm{T}}\mathbf{A}^{-1}\mathbf{b}\right)$ and 
    the equality $\left(\mathbf{\Sigma}_1^{-1}+\mathbf{\Sigma}_2^{-1}\right)^{-1} = \mathbf{\Sigma}_1\left(\mathbf{\Sigma}_1+\mathbf{\Sigma}_2\right)^{-1}\mathbf{\Sigma}_2 = \mathbf{\Sigma}_2\left(\mathbf{\Sigma}_1+\mathbf{\Sigma}_2\right)^{-1}\mathbf{\Sigma}_1$. The last two equalities follow by the Woodbury matrix identity $\displaystyle \left(\mathbf{A}+\mathbf{U}\mathbf{C}\mathbf{V}\right)^{-1}=\mathbf{A}^{-1}-\mathbf{A}^{-1}\mathbf{U}\left(\mathbf{C}^{-1}+\mathbf{V}\mathbf{A}^{-1}\mathbf{U}\right)^{-1}\mathbf{V}\mathbf{A}^{-1}$. Then the Lemma~\ref{lemma::Squared Hellinger Divergence between two Gaussian Distribution} follows by direct calculation.
\end{proof}

\begin{lemma}
    \label{lemma::upper bound of the second moment of gradient-difference-norm}
    Suppose the gradient of loss function is $L$-Lipschitz, \ie, $\displaystyle\sup_{w\in\mathcal{W}, z\in\mathcal{Z}}\|\nabla_{w}\ell(w,z)\|\leq L$. Then, we have
    \begin{equation}
        \Expectation{W,Z_i^\pm}{\|\nabla_{W}\ell(W,Z_i^+) - \nabla_{W}\ell(W,Z_i^-)\|^2} \leq \frac{L^2}{2} .
    \end{equation}
\end{lemma}

\begin{proof}
    Since $\nabla_{W}\ell(W,Z_i^+)$ and $\nabla_{W}\ell(W,Z_i^-)$ are independent and identically distributed conditioned on $W$, we have
    \begin{align}
        \Expectation{W,Z_i^\pm}{\|\nabla_{W}\ell(W,Z_i^+) - \nabla_{W}\ell(W,Z_i^-)\|^2}
        &= 2 \Expectation{W}{\Expectation{Z|W}{\|\nabla_{W}\ell(W,Z)\|^2} - \left(\Expectation{Z|W}{\|\nabla_{W}\ell(W,Z)\|}\right)^2} \\
        &= 2\Expectation{W}{\mathrm{Var}_{Z|W}(\|\nabla_{W}\ell(W,Z)\|)} \\
        &\leq L^2 / 2, 
    \end{align}
    where we remove the superscript and subscript of $Z$ for simplicity. The last inequality follows since the variance of a bounded random variable $X\in[a, b]$ is no more than $(b-a)^2 / 4$. 
\end{proof}
Now we are ready to prove the main result of this subsection.
\begin{proof}[Proof of Theorem~\ref{theorem::f-ICIMI bounds on SGLD::chain rule}]
    By Theorem~\ref{theorem::f-ICIMI}, we have
    \begin{align}
        \InDistributionGeneralizationGap
        &\leq \frac{b-a}{n}\sum_{i=1}^n\sqrt{\frac{2\fCMI[f]{W_T}{R_i}{Z_i^{\pm}}}{f''(1)}} \\
        & \leq \frac{b-a}{n}\sum_{i=1}^n\sqrt{\frac{2\fCMI[f]{W_{[0:T]}}{R_i}{Z_i^{\pm}}}{f''(1)}} \\
        &\leq \frac{2(b-a)}{n}\sum_{i=1}^n\sum_{t\in \mathcal{T}_i}\sqrt{\fCMI[H^2]{W_t}{R_i}{W_{t-1},Z_i^\pm}},
        \label{equation::upper bound of SGLD in terms of chain rule of squared Hellinger divergence}
    \end{align}
    where the last inequality follows by Lemma~\ref{lemma::Chain rule of squared Hellinger divergence}. As a consequence, it is sufficient to bound the $\fCMI[H^2]{W_t}{R_i}{W_{t-1},Z_i^\pm}$ term.
    Let $\ProbabilityKernel{R_i}{W_{t-1},Z_i^\pm} = \alpha\delta_1 + (1-\alpha)\delta_{-1}$, \ie, $R_i$ takes 1 with probability $\alpha$ and takes $-1$ with probability $(1-\alpha)$, given $W_{t-1}$ and $Z_i^\pm$. On the other hand, we have
    \begin{equation}
        \ProbabilityKernel{W_t}{W_{t-1},Z_i^\pm, R_i} = 
        \begin{cases}
            \mathcal{N}\left(W_t - \eta_t\nabla_W\ell(W_{t-1},Z_i^+), \sigma_t^2\mathbf{I}\right)\overset{\triangle}{=}\mathcal{N}^+, &\text{if} \ R_i=1,\\
            \mathcal{N}\left(W_t - \eta_t\nabla_W\ell(W_{t-1},Z_i^-), \sigma_t^2\mathbf{I}\right)\overset{\triangle}{=}\mathcal{N}^-, &\text{if} \ R_i=-1,
        \end{cases}
        \label{equation::W_t given W_t-1 and Z_i and R_i is Gaussian}
    \end{equation}
    and thus $\ProbabilityKernel{W_t}{W_{t-1},Z_i^\pm}$ is mixture Gaussian, expressed by
    \begin{equation}
        \ProbabilityKernel{W_t}{W_{t-1},Z_i^\pm} = \alpha\mathcal{N}^+ + (1-\alpha)\mathcal{N}^-. 
        \label{equation::W_t given W_t-1 and Z_i is mixture Gaussian}
    \end{equation}
    Keep in mind that $\alpha$ is a function of $W_{t-1}$ and $Z_i^\pm$, and the Gaussian components $\mathcal{N}^+$ and $\mathcal{N}^-$ are also parameterized by $W_{t-1}$ and $Z_i^\pm$. We drop these dependencies from the notation for abbreviation.
    Plugging~\eqref{equation::W_t given W_t-1 and Z_i and R_i is Gaussian} and~\eqref{equation::W_t given W_t-1 and Z_i is mixture Gaussian} into~\eqref{equation::upper bound of SGLD in terms of chain rule of squared Hellinger divergence} yields
    \begin{align}
        \fCMI[H^2]{W_t}{R_i}{W_{t-1},Z_i^\pm} &= \Expectation{W_{t-1},Z_i^\pm}{\Expectation{R_i|W_{t-1},Z_i^\pm}{\SquaredHellingerDivergence{\ProbabilityKernel{W_t}{W_{t-1},Z_i^\pm,R_i}}{\ProbabilityKernel{W_t}{W_{t-1},Z_i^\pm}}}}\\
        &= \Expectation{W_{t-1},Z_i^\pm}{\alpha\SquaredHellingerDivergence{\mathcal{N}^+}{\alpha\mathcal{N}^++(1-\alpha)\mathcal{N}^-} + (1-\alpha)\fDivergence{\mathcal{N}^-}{\alpha\mathcal{N}^++(1-\alpha)\mathcal{N}^-}} \\
        &\leq \Expectation{W_{t-1},Z_i^\pm}{\alpha(1-\alpha)\left(\SquaredHellingerDivergence{\mathcal{N}^+}{\mathcal{N}^-} + \SquaredHellingerDivergence{\mathcal{N}^-}{\mathcal{N}^+}\right)}\\
        &\leq \Expectation{W_{t-1},Z_i^\pm}{\frac{1}{4}\left(\SquaredHellingerDivergence{\mathcal{N}^+}{\mathcal{N}^-} + \SquaredHellingerDivergence{\mathcal{N}^-}{\mathcal{N}^+}\right)}.
    \end{align}
    To obtain the first two inequalities, we use the fact that the $f$-divergence $\fDivergence{\cdot}{\cdot}$ is jointly convex \wrt its arguments, and thus separately convex. To obtain the last inequality, we use the fundamental inequality $\sqrt{xy}\leq \frac{x+y}{2}$ to eliminate $\alpha$.
    Now, we choose $D_f$ to be the squared Hellinger divergence. By Lemma~\ref{lemma::Squared Hellinger Divergence between two Gaussian Distribution}, we have
    \begin{align}
        \SquaredHellingerDivergence{\mathcal{N}^+}{\mathcal{N}^-} 
        &= \SquaredHellingerDivergence{\mathcal{N}^-}{\mathcal{N}^+} \\
        &= 2 -2\exp\left(-\frac{\eta_t^2\|\nabla_W\ell(W_{t-1},Z_i^+) - \nabla_W\ell(W_{t-1},Z_i^-)\|^2}{8\sigma_t^2}
        \right),
    \end{align}
    and thus 
    \begin{align}
    \fCMI[H^2]{W_t}{R_i}{W_{t-1},Z_i^\pm}
        &\leq \Expectation{W_{t-1},Z_i^\pm}{1 - \exp\left(-\frac{\eta_t^2\|\nabla_W\ell(W_{t-1},Z_i^+) - \nabla_W\ell(W_{t-1},Z_i^-)\|^2}{8\sigma_t^2}
        \right)} \\
        &\leq 1 - \exp\left(-\frac{\eta_t^2}{8\sigma_t^2}\Expectation{W_{t-1},Z_i^\pm}{\|\nabla_W\ell(W_{t-1},Z_i^+) - \nabla_W\ell(W_{t-1},Z_i^-)\|^2}\right) \\
        &\leq 1 - \exp\left(-\frac{\eta_t^2L^2}{16\sigma_t^2}\right).
        \label{equation::upper bound of f mutual information between W_t and R_i conditioned on W_t-1 and Z_i::Chain rule of H^2}
    \end{align}
To obtain the second inequality we use the concavity of the function $1 - e^{-x}$, and to obtain the last inequality we use Lemma~\ref{lemma::upper bound of the second moment of gradient-difference-norm} and the fact that $1 - e^{-x}$ is monotonously increasing. Plugging~\eqref{equation::upper bound of f mutual information between W_t and R_i conditioned on W_t-1 and Z_i::Chain rule of H^2} into~\eqref{equation::upper bound of SGLD in terms of chain rule of squared Hellinger divergence} yields
\begin{align}
    \InDistributionGeneralizationGap &\leq \frac{2(b-a)}{n}\sum_{i=1}^n\sum_{t\in\mathcal{T}_i}\sqrt{1 - \exp\left(-\frac{\eta_t^2L^2}{16\sigma_t^2}\right)} \\
    &= \frac{2(b-a)}{n}\sum_{t\in [T]} \sqrt{1 - \exp\left(-\frac{\eta_t^2L^2}{16\sigma_t^2}\right)}, \nonumber
\end{align}
which is the desired result.
\end{proof}

\subsection{Proof of Lemma~\ref{lemma::subadditivity of f mutual information}}
\label{subsection::appendix::proof of lemma::subadditivity of f-divergence}

\begin{proof}[Proof of Lemma~\ref{lemma::subadditivity of f mutual information}]
    \begin{figure*}[ht]
            \centering
            \includegraphics[width=0.5\textwidth]{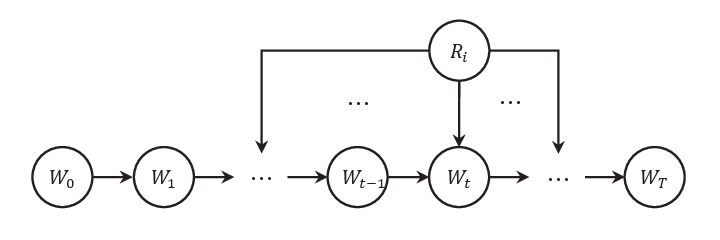}
            \caption{The Bayes network structure of SGLD algorithm (given $Z_i^\pm = z_i^\pm$ fixed and only $R_i$ is involved).}
            \label{figure::Bayes network of SGLD}
    \end{figure*}
    
    Let the super-sample $Z_i^\pm = z_i^\pm$ be fixed. 
    The SGLD optimization trajectory induces a Bayes network $\mathcal{G}$, as illustrated in Fig.~\ref{figure::Bayes network of SGLD}, where only $R_i$ is considered. Since $D_f$ is sub-additive, we have
    \begin{align}
        \fDivergence{P}{Q} 
        &\leq \sum_{v\in\mathcal{G}}\fDivergence{P_{\{v\}\cup\Pi_v}}{Q_{\{v\}\cup\Pi_v}} \\
        &= \fDivergence{P_{R_i}}{Q_{R_i}} + \fDivergence{P_{W_0}}{Q_{W_0}} + \sum_{t\in\mathcal{T}_i}\fDivergence{P_{W_t,W_{t-1},R_i}}{Q_{W_t,W_{t-1},R_i}} + \sum_{t\notin\mathcal{T}_i}\fDivergence{P_{W_t,W_{t-1}}}{Q_{W_t,W_{t-1}}},
        \label{equation::f divergence between P and Q can be decomposed into 3 parts}
    \end{align}
    for all distributions $P$ and $Q$ over $ \mathcal{W}^{T+1}\times \{-1, +1\}$. Now define $P$ as the joint distribution of $W_{[0:T]}$ and $R_i$ induced by the SGLD algorithm, \ie,
    \begin{align}
        P &\coloneqq P_{R_i} \otimes \ProbabilityKernel{W_{[0:T]}}{R_i} \\
        &= P_{R_i}P_{W_0}\prod_{t\notin\mathcal{T}_i}\ProbabilityKernel{W_t}{W_{t-1}}\prod_{t\in\mathcal{T}_i}\ProbabilityKernel{W_t}{W_{t-1},R_i},
    \end{align}
    and define $Q$ as the product of marginal distribution of $W_{[0:T]}$ and $R_i$, \ie,
    \begin{align}
        Q &\coloneqq P_{R_i} \otimes P_{W_{[0:T]}} \\
        &= P_{R_i}P_{W_0}\prod_{t\in[T]}\ProbabilityKernel{W_t}{W_{t-1}}.
    \end{align}
    By induction over $t$, one can prove that $P$ and $Q$ have the same marginal distribution at each vertex in the graph $\mathcal{G}$, \ie, $P_{R_i} = Q_{R_i}$, $P_{W_0} = Q_{W_0}$, and $P_{W_t} = Q_{W_t},\ \forall t\in [T]$. Therefore, we have 
    \begin{subequations}
        \begin{align}
            \fDivergence{P_{R_i}}{Q_{R_i}} &= 0, \\
            \fDivergence{P_{W_0}}{Q_{W_0}} &= 0, \\
            \begin{split}
                \fDivergence{P_{W_t,W_{t-1}}}{Q_{W_t,W_{t-1}}} 
                &= \fDivergence{P_{W_{t-1}}\ProbabilityKernel{W_t}{W_{t-1}}}{Q_{W_{t-1}}Q_{W_t|W_{t-1}}} \\
                &= 0, \ t\notin\mathcal{T}_i, 
            \end{split}\\
            \begin{split}
                \fDivergence{P_{W_t,W_{t-1},R_i}}{Q_{W_t,W_{t-1},R_i}} 
                &= \Expectation{W_{t-1}}{\fDivergence{\ProbabilityKernel{W_t, R_i}{W_{t-1}}}{Q_{W_t,R_i|W_{t-1}}}},\\ 
                & = \Expectation{W_{t-1}}{\fDivergence{\ProbabilityKernel{W_t, R_i}{W_{t-1}}}{\ProbabilityKernel{W_t}{W_{t-1}}P_{R_i}}},\ t\in\mathcal{T}_i.
            \end{split}
        \end{align}
    \end{subequations}
    Add all terms together and use the fact that $\fDivergence{P}{Q} = I_f(W_{[0:T]};R_i)$, we have
    \begin{equation}
        I_f(W_{[0:T]};R_i)\leq \sum_{t\in\mathcal{T}_i}\Expectation{W_{t-1}}{\fDivergence{\ProbabilityKernel{W_t, R_i}{W_{t-1}}}{\ProbabilityKernel{W_t}{W_{t-1}}P_{R_i}}}.
        \label{equation::I_f(W^T;R_i) conditioned on Z = z}
    \end{equation}
    Keep in mind that~\eqref{equation::I_f(W^T;R_i) conditioned on Z = z} is established based on the condition $Z_i^\pm = z_i^\pm$. Thus, the desired result~\eqref{equation::subadditivity of f mutual information} follows by taking expectation over $z_i^\pm$.
\end{proof}


\bibliographystyle{IEEEtran}

\end{document}